\documentclass[11pt,a4paper]{article}

\newif\iffull
\fulltrue

\usepackage{amsmath,amsfonts,amssymb,amsthm}
\usepackage{mathbbol}
\usepackage{amsthm}
\usepackage{graphicx,color}
\usepackage{boxedminipage}
\iffull
\usepackage[linesnumbered, boxed, algosection,noline]{algorithm2e}
\usepackage[pdftex, plainpages = false, pdfpagelabels, 
                 bookmarks=false,
                 bookmarksopen = true,
                 bookmarksnumbered = true,
                 breaklinks = true,
                 linktocpage,
                 pagebackref,
                 colorlinks = true,  
                 linkcolor = blue,
                 urlcolor  = blue,
                 citecolor = red,
                 anchorcolor = green,
                 hyperindex = true,
                 hyperfigures
                 ]{hyperref} 
\usepackage{vmargin}
\setmarginsrb{1in}{1in}{1in}{1in}{0pt}{0pt}{0pt}{6mm}
\fi
\usepackage{breakurl}

\hypersetup{breaklinks=true}
\usepackage{framed}
\usepackage{thmtools}
\usepackage{thm-restate}
\usepackage{xspace}
\usepackage{todonotes}
\usepackage{xifthen}
\usepackage{tabularx}
\usepackage{tikz}
\usepackage{pgfplots}

 \usetikzlibrary{calc}

\newtheorem{theorem}{Theorem}
\newtheorem{proposition}{Proposition}
\newtheorem{lemma}{Lemma}

\newtheorem{claim}{Claim}
\newtheorem{corollary}{Corollary}

\newcommand{\pname}{\textsc}
\newcommand{\ProblemFormat}[1]{\pname{#1}}
\newcommand{\ProblemIndex}[1]{\index{problem!\ProblemFormat{#1}}}
\newcommand{\ProblemName}[1]{\ProblemFormat{#1}\ProblemIndex{#1}{}\xspace}

\newcommand{\probPCAOut}{\ProblemName{PCA with Outliers}}

\DeclareMathOperator{\rank}{rank}
\DeclareMathOperator{\sign}{sign}
\DeclareMathOperator{\proj}{proj}
\DeclareMathOperator{\vspan}{span}
\DeclareMathOperator{\comp}{comp}

\DeclareMathOperator{\operatorClassNP}{{\sf NP}}
\newcommand{\classNP}{\ensuremath{\operatorClassNP}}

\DeclareMathOperator{\operatorClassFPT}{{\sf FPT}\xspace}
\newcommand{\classFPT}{\ensuremath{\operatorClassFPT}\xspace}
\DeclareMathOperator{\operatorClassW}{{\sf W}}
\newcommand{\classW}[1]{\ensuremath{\operatorClassW[#1]}}

\newlength{\RoundedBoxWidth}
\newsavebox{\GrayRoundedBox}
\newenvironment{GrayBox}[1]%
   {\setlength{\RoundedBoxWidth}{.93\linewidth}
    \def\boxheading{#1}
    \begin{lrbox}{\GrayRoundedBox}
       \begin{minipage}{\RoundedBoxWidth}}%
   {   \end{minipage}
    \end{lrbox}
    \begin{center}
    \begin{tikzpicture}%
       \node(Text)[draw=black!20,fill=white,rounded corners,%
             inner sep=2ex,text width=\RoundedBoxWidth]%
             {\usebox{\GrayRoundedBox}};
        \coordinate(x) at (current bounding box.north west);
        \node [draw=white,rectangle,inner sep=3pt,anchor=north west,fill=white] 
        at ($(x)+(6pt,.75em)$) {\boxheading};
    \end{tikzpicture}
    \end{center}}     

\newenvironment{defproblemx}[2][]{\noindent\ignorespaces%
                                \FrameSep=6pt%
                                \parindent=0pt%
                \vspace*{-1.5em}
                \ifthenelse{\isempty{#1}}{%
                  \begin{GrayBox}{\textsc{#2}}%
                }{%
                  \begin{GrayBox}{\textsc{#2} parameterized by~{#1}}%
                }
                \begin{tabular*}{\linewidth}{@{\hspace{.1em}} >{\itshape} p{0.1\linewidth} p{0.8\linewidth} @{}}%
            }{
                \end{tabular*}%
                \end{GrayBox}%
                \ignorespacesafterend
            }  

\newcommand{\defparproblema}[4]{
  \begin{defproblemx}[#3]{#1}
    Input:  & #2 \\
    Task: & #4
  \end{defproblemx}
}%

\newcommand{\defproblema}[3]{
  \begin{defproblemx}{#1}
    Input:  & #2 \\
    Task: & #3
  \end{defproblemx}
}%

\newcommand{\Oh}{\mathcal{O}}

\begin{document}
\iffull
    \title{Refined Complexity of PCA with Outliers\thanks{This work is supported by the Research Council of Norway via the project ``MULTIVAL''.}}
\author{
Fedor V. Fomin\thanks{
    Department of Informatics, University of Bergen, Norway, \texttt{\{fedor.fomin, petr.golovach, fahad.panolan, kirill.simonov\}@uib.no}.
} \addtocounter{footnote}{-1}
\and
Petr A. Golovach\footnotemark{} \addtocounter{footnote}{-1}
\and 
Fahad Panolan\footnotemark{} \addtocounter{footnote}{-1}
\and 
Kirill Simonov\footnotemark{}
}
\maketitle
\fi

\begin{abstract}
Principal component  analysis  (PCA)  is one of the most fundamental procedures in  exploratory data analysis and is the basic step in   applications ranging from    quantitative finance and  bioinformatics to image analysis and neuroscience. 
However, it is well-documented that the  applicability of PCA in many real scenarios could be  constrained by an ``immune deficiency''  to outliers such as  corrupted observations. We consider the following algorithmic question about the PCA with outliers. For a 
set of $n$ points in $\mathbb{R}^{d}$, how to learn a subset of  points, say 1\%  of the total number of points, such that 
the remaining  part of the points is best fit into some unknown  $r$-dimensional subspace? We provide a rigorous algorithmic analysis of the problem. We show that the problem is solvable in time  $n^{\Oh(d^2)}$. In particular, for constant dimension the problem is solvable in polynomial time. We complement the algorithmic result by the lower bound, showing that unless Exponential Time Hypothesis fails, in time  $f(d)n^{o(d)}$, for any function $f$ of $d$, it is impossible not only to solve the problem exactly but even to approximate it within a constant factor. 
\end{abstract}

\section{Introduction}
\medskip
\noindent\textbf{Problem statement and motivation.} 
Classical \emph{principal component analysis} (PCA) is 
one of the most popular and successful techniques used for dimension reduction  in data analysis and machine learning \cite{pearson1901liii,hotelling1933analysis,eckart1936approximation}. In PCA 
one seeks the best   low-rank approximation of data matrix $M$ by solving
\begin{eqnarray*}
\text{ minimize } \|M-L\|^2_F \\ \text{ subject to } \rank(L) \leq  r.
\end{eqnarray*}
Here $||A||_F^2 = \sum_{i, j} a_{ij}^2$ is the square of the Frobenius norm of matrix $A$.
By the Eckart-Young theorem \cite{eckart1936approximation},  PCA is efficiently solvable via
  Singular Value Decomposition (SVD). PCA is used as 
 a preprocessing step in a great variety of  modern applications including face recognition, data classification,     and analysis of  social networks.  

In this paper we consider a variant of PCA with outliers, where we wish to recover a low-rank matrix from large but sparse errors.
Suppose that we have $n$ points (observations) in $d$-dimensional space.  We know that a part of the   points are  arbitrarily located (say, produced by corrupted observations) while the remaining points are close to an $r$-dimensional true subspace. We do not have any information about the true subspace and about the corrupted observations. Our task is to learn the true subspace and to identify the outliers. As a common practice, we collect the points into $n\times d$ matrix $M$, thus each of the rows of $M$ is a point   and the columns of  
$M$ are the coordinates. 
However,  it is very likely that  PCA of $M$ will not reveal any reasonable information about non-corrupted observations---well-documented drawback  of PCA is its vulnerability to even very small number of outliers, an example is shown in Figure~\ref{fig:badPCA}. 

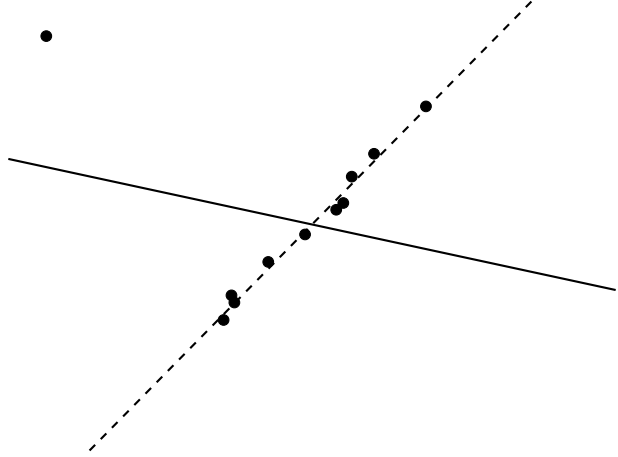
\begin{figure}[ht]
    \centering
        \begin{tikzpicture}[scale=1]
            \filldraw (-0.5763, -0.4971) circle (2pt);
            \filldraw (-1.0606, -0.9408) circle (2pt);
            \filldraw (-0.0914, -0.1326) circle (2pt);
            \filldraw (0.4123, 0.2850) circle (2pt);
            \filldraw (1.5009, 1.5670) circle (2pt);
            \filldraw (-1.0224, -1.0361) circle (2pt);
            \filldraw (0.8172, 0.9383) circle (2pt);
            \filldraw (-1.1646, -1.2677) circle (2pt);
            \filldraw (0.3191, 0.1952) circle (2pt);
            \filldraw (0.5238, 0.6348) circle (2pt);
            \filldraw (-3.5000, 2.5000) circle (2pt);
            
            \draw[thick, dashed] (-2.93,-3)--(2.93,3);
            \draw[thick] (-4,0.87)--(4,-0.87);

        \end{tikzpicture}
        \caption{An illustration on how outliers impact PCA. The optimal approximation line (in dashed) of the given set of points without the evident outlier shows the linear structure of the dataset. However, when the outlier is present, the principal component (in solid) changes drastically.}
    \label{fig:badPCA}
\end{figure}


%
%
 
Matrix formulation   suggests the following  interpretation: we seek a low-rank matrix $L$ that, with an exception in few rows,  approximates $M$  best.  
%
%
%
%

\medskip
\defproblema{\probPCAOut}%
{Data matrix $M\in \mathbb{R}^{n \times d}$, integer parameters $r$ and $k$.}%
{ 
\begin{eqnarray*}
\text{ minimize }& & \|M-L-S \|^2_F \\ 
\text{ subject to } 
& &  L,S\in \mathbb{R}^{n \times d}, \\
& &\rank(L) \leq  r, \text{ and}\\
          & & S \text{ has at most } k \text{ non-zero rows}. 
\end{eqnarray*}
}
The geometric interpretation of \probPCAOut is very natural:  Given $n$ points in $\mathbb{R}^{d}$, we seek for a set of $k$ points whose removal leaves  the remaining $n-k$ points as close as possible  to some $r$-dimensional subspace.   


\medskip
\noindent\textbf{Related work.} 
 \probPCAOut  belongs to the large class of  
extensively studied robust PCA problems,  see e.g. \cite{VaswaniN18,XuCS10,bouwmans2016handbook}.  
   In the robust PCA setting we observe a noisy version  $M$  of data  matrix $L$ whose principal components we have to discover. 
   In the case when  $M$ is a ``slightly'' disturbed version of $L$,  PCA performed on $M$ provides a reasonable approximation for $L$. However, when $M$ is very ``noisy'' version of $L$,    like being corrupted  by  a few outliers, 
   even one corrupted outlier can arbitrarily alter the quality of the approximation. 

One of the approaches to robust PCA, which is relevant to our work,  is to model outliers as additive sparse matrix. Thus 
we have a data $d\times n$ matrix $M$, which is the superposition of a low-rank component $L$ and a sparse component $S$. 
That is, 
$M=L+S.$
This  approach became popular after the works of  
Cand{\`{e}}s et al.   \cite{CandesLMW11}, Wright et al.  \cite{WrightGRPM09}, and 
 Chandrasekaran et al. \cite{ChandrasekaranSPW11}. A significant body of work on the robust PCA  problem has been centered around proving that, under some feasibility assumptions on $M$, $L$, and $S$, a solution to 
 \begin{eqnarray}\label{eq_optpr1}
\text{ minimize }& & \rank(L) + \lambda\| S\| _{0} \\ 
\text{ subject to }  & &  M=L+S,  \nonumber
\end{eqnarray}
 where $\| S\| _{0}$ denotes the number of non-zero columns in matrix $S$  and $\lambda$ is a regularizing parameter, recovers matrix $L$ uniquely. While  optimization problem \eqref{eq_optpr1}  is \classNP-hard  \cite{GillisV18},  
it is possible to show that   under certain assumptions on $L$ and $S$, its convex  relaxation can recover these matrices  efficiently.  

The problem strongly related to  \eqref{eq_optpr1}  was studied in computational complexity under the name \textsc{Matrix Rigidity}  \cite{grigoRigid,grigoRigidTra,Valiant1977}. Here, for a given matrix $M$, and integers $r$ and $k$, the task is to decide whether at most $k$ entries of $M$ can be changes so that the rank of the resulting matrix is at most $r$. Equivalently, this is the problem to decide whether a given matrix $M=L+S$, where $\rank(L)\leq r$ and  $\| S\| _{0}\leq k$. 
Fomin et al. \cite{FominLMSZ18} gave an algorithm solving \textsc{Matrix Rigiity} in time 
$2^{\Oh(r \cdot k \cdot \log (r\cdot k))} \cdot (nd)^{\Oh(1)}$.
On the other hand, they show that the  problem is \classW1-hard parameterized by $k$. In particular, this implies that an algorithm of running time  $f(k)  \cdot (rnd)^{\Oh(1)}$ for this problem is highly unlikely for any function $f$ of $k$ only.
 
 A natural extension of the robust PCA approach  \eqref{eq_optpr1}  is to consider the noisy version of robust PCA: Given $M=L+S+N$, where $L$, $S$,  and $N$ are unknown, but $L$ is known to be low rank, $S$ is known to have a few non-zero rows, and noise matrix $N$ is of small Frobenius norm,   recover $L$. Wright et al. \cite{WrightGRPM09} studied the following model of noisy robust PCA:
\begin{eqnarray}\label{eq:prob2}
\text{ minimize }& & \rank(L) + \lambda\| S\| _{0} \\ 
\text{ subject to } 
& &  \|M-L-S\|_F^{2}\leq \varepsilon. \nonumber
\end{eqnarray}
 Thus
\eqref{eq:prob2} models the situations when  we want to learn the principal components of $n$ points in $d$-dimensional space under the assumption that a small number of coordinates is corrupted. 
 



The study of the natural, and seemingly more difficult extension of \eqref{eq_optpr1}  to the PCA with outliers, was initiated by 
Xu et al. \cite{XuCS10}, who introduced the following  idealization of the problem.  
\begin{eqnarray}\label{eq:prob3}
\text{ minimize }& & \rank(L) + \lambda\| S\| _{0,r}\\ 
\text{ subject to }  \nonumber
& &  M=L+S. 
\end{eqnarray}
Here $\| S\| _{0,r}$ denotes the number of non-zero columns in matrix $S$ and $\lambda$ is a regularizing parameter.  Xu et al. \cite{XuCS10}  approached this problem by building its  convex surrogate and applying efficient convex optimization-based algorithm for the surrogate.  Chen et al. \cite{chen2011robust} studied the variant of the problem with the partially observed data. 
Similar as \eqref{eq:prob2} is the noisy version of the robust PCA model \eqref{eq_optpr1}, 
the \probPCAOut problem studied in our work can be seen as a noisy version of \eqref{eq:prob3}.

\medskip
\noindent\textbf{Our results.}  
Even though \probPCAOut  was assumed to be NP-hard,  to the best of our knowledge, this has never been studied formally.  
While  NP-hardness  is a serious strike against the  tractability of the problem,  on the other hand, it only says that in the worst case the problem is not tractable. But since the complexity of the problem could be governed  by several parameters like the rank $r$ of $L$, the number of outliers $k$ or dimension $d$ of $M$, it is natural to ask how these parameters influence the complexity of the problem. For example, when $k$ is a small constant, we can guess which points are outliers and run PCA for the remaining points. This will bring us to $n^k$ calls of PCA which is polynomial for constant $k$ and is exponential when $k$ is a fraction of $n$. 

In this paper we give an algorithm 
 solving \probPCAOut roughly in time  $|M|^{\Oh(d^2)}$, where $|M|$ is the size of the input matrix $M$. Thus for fixed dimension $d$,   the problem is solvable  in polynomial time. 
 The algorithms works in polynomial time for any number of outliers $k$ and the rank $r$ of the recovered matrix $L$.
 Our algorithm strongly relies on the tools developed in computational algebraic geometry, in particular, for handling arrangements of algebraic surfaces in  $\mathbb{R}^d$ defined by polynomials of bounded degree. 

We complement our algorithmic result by a  complexity lower bound. Our lower bound not only implies that the problem is \classNP-hard when dimension $d$ is part of the input, it also rules out a possibility of certain type of algorithms for \probPCAOut. 
More precisely, assuming the Exponential Time Hypothesis (ETH),\footnote{ETH of Impagliazzo,  Paturi,   and Zane \cite{ImpagliazzoPZ01} is that 3-SAT with $n$-variables is not solvable in time $2^{o(n)}$.}  we show that for any constant $\omega \ge 1$,     \probPCAOut  cannot be $\omega$-approximated in time $f(d)|M|^{o(d)}$, for any function $f$ of $d$ only. 

Our algorithm is, foremost, of theoretical interest, especially in the presense of the nearly-matching lower bound showing that doing something essentially better is next to impossible. In practice, PCA is often applied to reduce high-dimensional datasets, and for this task the running time exponential in $d$ is not practical.
However, there are cases where such an algorithm could still be useful. One example could be the visualization of low-dimensional data, where the number of dimensions, even if it is small already, needs to be lowered down to two to actually draw the dataset. Another example could be when we suspect a small subset of features to be highly correlated, and we want to reduce them to one dimension in order to get rid of the redundancy in data. This potential application is well illustrated by the popular PCA tutorial \cite{Shlens14}, where essentially one-dimensional movement of a spring-mass is captured by three cameras, resulting in 6 features.

\section{Polynomial  algorithm for bounded dimension}\label{sec:alg}
\subsection{Preliminaries}

As a subroutine in our algorihm, we use a standard result about sampling points from cells of an arrangement of algebraic surfaces, so first we state some definitions and an algorithm from \cite{BasuPR06}.

We denote the ring of polynomials in variables $X_1$, \dots, $X_d$ with coefficients in $\mathbb{R}$ by $\mathbb{R}[X_1, \dots, X_d]$. By saying that an algebraic set $V$ in $\mathbb{R}^d$ is defined by $Q \in \mathbb{R}[X_1, \dots, X_d]$,  we mean that $V = \{x \in \mathbb{R}^d | Q(x_1, \dots, x_d) = 0\}$.
For a set of $s$ polynomials $\mathcal{P} = \{P_1, \dots, P_s\} \subset \mathbb{R}[X_1, \dots, X_d]$, a sign condition is specified by a sign vector $\sigma \in \{-1, 0, +1\}^s$, and the sign condition is non-empty over $V$ with respect to $\mathcal{P}$ if there is a point $x \in V$ such that
$$\sigma = (\sign(P_1(x)), \dots, \sign(P_s(x))),$$
where $\sign(x)$ is the sign function on real numbers defined as
$$\sign(x) = \begin{cases}
    1, \text{ if } x > 0,\\
    0, \text{ if } x = 0,\\
    -1, \text{ if } x < 0
\end{cases}$$
for $x \in \mathbb{R}$.

The \emph{realization space} of $\sigma \in \{-1, 0, +1\}^s$ over $V$ is the set
$$R(\sigma) = \{x | x \in V, \sigma = (\sign(P_1(x)), \dots, \sign(P_s(x)))\}.$$
If $R(\sigma)$ is not empty then each of its non-empty semi-algebrically connected (which is equivalent to just connected on semi-algebraic sets as proven in \cite{BasuPR06}, Theorem 5.22) components is a \emph{cell} of $\mathcal{P}$ over $V$.

For an algebraic set $W$ its real dimension is the maximal integer $d'$ such that there is a homeomorphism of $[0, 1]^{d'}$ in $W$. Naturally, if $W \subset \mathbb{R}^d$, then $d' \le d$.

The following theorem from \cite{BasuPR06} gives an algorithm to compute a point in each cell of $\mathcal{P}$ over $V$.

\begin{proposition}[\cite{BasuPR06}, Theorem 13.22]
    Let $V$ be an algebraic set in $\mathbb{R}^d$ of real dimension $d'$ defined by $Q(X_1, \dots, X_d) = 0$, where $Q$ is a polynomial in $\mathbb{R}[X_1, \dots, X_d]$ of degree at most $b$, and let $\mathcal{P} \subset \mathbb{R}[X_1, \dots, X_d]$ be a finite set of $s$ polynomials with
    each $P \in \mathcal{P}$ also of degree at most $b$. Let $D$ be a ring generated by the coefficients of $Q$ and the polynomials in $\mathcal{P}$. There is an algorithm which takes as input $Q$, $d'$ and $\mathcal{P}$ and computes a set of points meeting every non-empty cell of $V$ over $\mathcal{P}$. The algorithm uses at most
    $s^{d'} b^{O(d)}$ arithmetic operations in $D$.
    \label{prop:sample_points}
\end{proposition}

On the practical side, we note that a number of routines from \cite{BasuPR06} is implemented in the SARAG library \cite{Caruso06}.

\subsection{Algorithm}

First, we emphasize on a folklore observation that geometrically the low-rank approximation matrix $L$ is defined as orthogonal projection of rows of $M$ on some $r$-dimensional subspace of $\mathbb{R}^d$. For the proof see e.g. \cite{BlumHK17}.

\begin{proposition}
    Given a matrix $M \in \mathbb{R}^{n \times d}$ with rows $m_1$, \ldots, $m_n$, the task of finding a matrix $L$ of rank at most $r$ which minimizes $||M - L||_F^2$ is equivalent to finding an $r$-dimensional subspace of $\mathbb{R}^d$ which minimizes the total squared distance from rows of $M$ treated as points in $\mathbb{R}^d$:
    \iffull
        $$\min_{\substack{L \in \mathbb{R}^{n \times d}\\\rank L \le r}} ||M - L||_F^2 = \min_{\substack{U \subset \mathbb{R}^d\\U \text{ is a linear subspace of } \dim r}} \sum_{i = 1}^n ||m_i - \proj_U m_i||_F^2,$$
    
    \fi
    where $\proj_U x$ is the orthogonal projection of $x$ on $U$ for $x \in \mathbb{R}^d$.
    \label{prop:subspace}
\end{proposition}
%

By Proposition \ref{prop:subspace}, if we fix an $r$-dimensional subspace $U$ containing the span of rows of $L$, then the outliers are automatically defined as $k$ farthest points from $U$ among $\{m_i\}_{i=1}^n$. In the next proposition, we give a precise statement of this.

\begin{proposition}
    The optimization objective of \probPCAOut for a given matrix $M \in \mathbb{R}^{n \times d}$ with rows $m_1$, \ldots, $m_n$ can be equivalently redefined as follows.
    \iffull
        $$\min_{\substack{L, S \in \mathbb{R}^{n \times d}\\\rank L \le r\\S \text{ has at most $k$ non-zero rows}}} ||M - L - S||_F^2 = \min_{\substack{U \subset \mathbb{R}^d\\U \text{ is a linear subspace of } \dim r}} ||M - L_U - S_U||_F^2,$$
    \fi
    where $S_U$ has $k$ non-zero rows which are $k$ rows of $M$ with the largest value of $||m_i - \proj_U m_i||_F^2$, and $L_U$ is the orthogonal projection of the rows of $(M - S_U)$ on $U$.
    $$S_U = \begin{pmatrix}
        m_1\\
        \vdots\\
        m_k\\
        0\\
        \vdots\\
        0
    \end{pmatrix}, \quad L_U = \begin{pmatrix}
        0\\
        \vdots\\
        0\\
        \proj_U m_{k + 1}\\
        \vdots\\
        \proj_U m_n
    \end{pmatrix},$$
    assuming that rows of $M$ are ordered by descending $||m_i - \proj_U m_i||_F^2$.
\label{prop:subspace_outliers}
\end{proposition}

So for a fixed $U$ we may determine $S_U$ easily and then solve the classical PCA for the matrix $(M - S_U)$. The intuition behind our algorithm is that the set of $k$ farthest points is the same for many subspaces, and solving PCA for $(M - S_U)$ treats all these subspaces. The crucial point is to bound the number of different matrices $S_U$ we have to consider. There is of course a trivial bound of $n^k$ since $S_U$ is always obtained by choosing $k$ rows of $M$.
But the number of different $S_U$ is also geometrically limited, and exploiting this we are able to obtain another bound of $n^{O(d^2)}$, resulting in the following theorem.

\begin{theorem}
    Solving \probPCAOut is reducible to solving
    $$\binom{n}{2}^{\min(rd, (d - r)d)} 2^{\Oh(d)} = n^{\Oh(d^2)}$$
    instances of PCA. This reduction can be computed in the number of operations over $\mathbb{R}$ bounded by the expression above.
    \label{thm:algorithm}
\end{theorem}

First, a note about the statement of Theorem~\ref{thm:algorithm}. Our algorithm relies on solving the classical PCA, and since only iterative algorithms for PCA and SVD exist, we could not claim that our algorithm solves \probPCAOut in some fixed number of operations. However, if we are only interested in solving the problem up to some constant precision, for example machine epsilon, then PCA is solvable in polynomial number of operations and so by Theorem~\ref{thm:algorithm}, \probPCAOut is solvable in $n^{\Oh(d^2)}$ operations.

\begin{proof}[Proof of Theorem~\ref{thm:algorithm}]

    We start with associating $r$-dimensional subspaces of $\mathbb{R}^d$ with points of a certain algebraic set. 
    Consider the matrix space $\mathbb{R}^{(d - r) \times d}$, and for an element $V \in \mathbb{R}^{(d - r) \times d}$, $V = \{v_{ij}\}_{i, j}$, the following polynomial conditions:
    \begin{gather*}
        Q^O_{i,j}(V) := \sum_{l = 1}^d v_{il} v_{jl} = 0, \text{ for }1 \le i < j \le (d - r),\\
        Q^N_{i}(V) := \left(\sum_{l = 1}^d v_{il}^2\right) - 1 = 0, \text{ for } 1 \le i \le (d - r),
    \end{gather*}
    where condition $Q^O_{i, j}(V) = 0$ requires rows $i$, $j$ of $V$ to be pairwise orthogonal and condition $Q^N_j(V) = 0$ requires row $j$ of $V$ to have length 1. We may write all these conditions as a single polynomial condition $Q(V) = 0$ by taking the sum of squares:
    $$Q(V) = \sum_{1 \le i < j \le (d - r)} (Q^O_{i, j}(V))^2 + \sum_{i = 1}^{d - r} (Q^N_i(V))^2.$$
    Thus $Q(V) = 0$ if and only if each of $Q^O_{i, j}(V)$ and each of $Q^N_i(V)$ is 0.
    
    Consider an algebraic set $W \subset \mathbb{R}^{(d - r) \times d}$ defined as the zero set of $Q(V)$. For any $V \in W$ with rows $v_1$, \ldots, $v_{d - r}$, consider the $r$-dimensional subspace $\comp(V) := \vspan(\{v_1, \cdots, v_{d - r}\})^\bot \subset \mathbb{R}^d$ which is the orthogonal complement of the span of the rows of $V$. Since $Q(V) = 0$, the rows of $V$ are pairwise orthogonal and are of length 1. Then, the dimension of $\comp(V)$ is $r$ and for any point $x \in \mathbb{R}^d$ the squared distance from $x$ to $\comp(V)$ is equal to
    $$\sum_{i = 1}^{d - r} (v_i \cdot x)^2 = ||V x^T||_F^2,$$
    assuming that $x$ is a row vector.

    Each $V \in W$ defines an $r$-dimensional subspace $\comp(V) \subset \mathbb{R}^d$ and each $r$-dimensional subspace $U \subset \mathbb{R}^d$ is of this form for some $V \in W$ since there exists an orthonormal basis of the orthogonal complement of $U$. Then we can reformulate Proposition \ref{prop:subspace_outliers} in terms of elements of $W$ as follows.
\iffull
    \begin{equation}
        \min_{\substack{L, S \in \mathbb{R}^{n \times d}\\\rank L \le r\\S \text{ has at most $k$ non-zero rows}}} ||M - L - S||_F^2 = \min_{V \in W} ||M - S_{\comp(V)} - L_{\comp(V)}||_F^2,
    \label{eq:objV}
    \end{equation}
\fi
where $S_{\comp(V)}$ and $L_{\comp(V)}$ are defined in accordance with notation in Proposition \ref{prop:subspace_outliers}. Let $m_1$, \ldots, $m_n$ be the rows of the input matrix $M$; $S_{\comp(V)}$ has $k$ non-zero rows which are $k$ rows of $M$ with the largest value of $||m_i - \proj_{\comp(V)} m_i||_F^2 = ||V m_i^T||_F^2$, and $L_{\comp(V)}$ is the orthogonal projection of the rows of $(M - S_{\comp(V)})$ on $\comp(V)$. Denote $S_{\comp(V)}$ by $S_V$ and $L_{\comp(V)}$ by $L_V$.

    Now, consider the set of polynomials $\mathcal{P} = \{P_{i, j}\}_{1 \le i < j \le n}$ defined on $W$, where
    $$P_{i, j}(V) = ||V m_i^T||_F^2 - ||V m_j^T||_F^2.$$

    Consider the partition $\mathcal{C}$ of $W$ on cells over $\mathcal{P}$. For each cell $C$, the sign condition with respect to $\mathcal{P}$ is constant over $C$, meaning that for every pair $1 \le i < j \le n$, the sign of
    $$||V m_i^T||_F^2 - ||V m_j^T||_F^2$$
    is the same for all $V \in C$. So the relative order on $\{||V m_i^T||_F^2\}_{i=1}^n$ is also the same for all $V \in C$. Since $||V m_i^T||_F^2$ is exactly the squared distance from $m_i$ to $V$, $k$ rows of $M$ which are the farthest are also the same for all $V \in C$. Then $S_{V}$ is constant over $V \in C$, denote this common value as $S_C$. We can rewrite \eqref{eq:objV} as
    \iffull
    \begin{equation*}
        \min_{V \in W} ||M - S_{V} - L_{V}||_F^2 = \min_{C \in \mathcal{C}} \min_{V \in C} ||M - S_{V} - L_{V}||_F^2 = \min_{C \in \mathcal{C}} \min_{V \in C} ||(M - S_C) - L_{V}||_F^2.
    \end{equation*}
    \fi
    Note that 
\iffull
    \begin{equation}
        \min_{C \in \mathcal{C}} \min_{V \in C} ||(M - S_C) - L_{V}||_F^2 = \min_{C \in \mathcal{C}} \min_{V \in W} ||(M - S_C) - L_{V}||_F^2,
        \label{eq:objC}
    \end{equation}
\fi
    as for any $C \in \mathcal{C}$, $\min_{V \in C} ||(M - S_C) - L_{V}||_F^2 \ge \min_{V \in W} ||(M - S_C) - L_{V}||_F^2$ since $C \subset W$. Also, any ($S_C$, $L_{V}$) in the right-hand side of \eqref{eq:objC} is still a valid choice of ($S$, $L$) for the original problem, and the optimum of the original problem is equal to the left-hand side of \eqref{eq:objC}. 

    For a fixed $C \in \mathcal{C}$ computing right-hand side of \eqref{eq:objC} is equivalent to solving an instance $(M - S_C, r)$ of the classical PCA by Proposition~\ref{prop:subspace}:
    \iffull
    $$\min_{V \in W} ||(M - S_C) - L_{V}||_F^2 = \min_{L \in \mathbb{R}^{n \times d},\, \rank(L) \le r} ||(M - S_C) - L||_F^2.$$
    \fi
    By the reasoning above, the optimum of the original instance of \probPCAOut is reached on one of the constructed instances $\{(M - S_C, r)\}_{C \in \mathcal{C}}$ of PCA. A toy example of an algebraic set $W$ and its partitioning is shown in Figure~\ref{fig:solspace}.

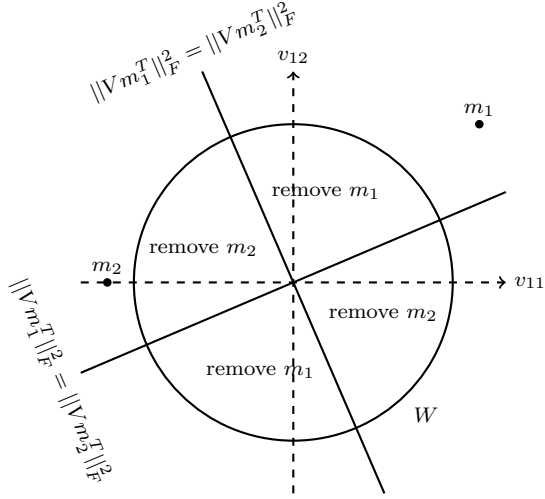
\begin{figure}[ht]
    \centering
    \scriptsize
    \begin{tikzpicture}[scale=.7]
            \draw[->,thick, dashed] (-4,0)--(4,0) node[right]{$v_{11}$};
            \draw[->,thick, dashed] (0,-4)--(0,4) node[above]{$v_{12}$};
            \draw[thick] (0, 0) circle (3cm);
            \node at (2.5, -2.5) () {$W$};
            \filldraw (3.5, 3) circle (2pt) node[above] {$m_1$};
            \filldraw (-3.5, 0) circle (2pt) node[above] {$m_2$};
            \draw[thick] (-4,-12/7)--(4,12/7) node[pos=0,below, rotate=-67]{$||V m_1^T||_F^2 = ||V m_2^T||_F^2$};
            \draw[thick] (-12/7,4)--(12/7,-4) node[pos=0,above, rotate=23]{$||V m_1^T||_F^2 = ||V m_2^T||_F^2$};
            \node[align=center] at (70:1.8) () {remove $m_1$};
            \node[align=center] at (250:1.8) () {remove $m_1$};
            \node[align=center] at (160:1.8) () {remove $m_2$};
            \node[align=center] at (340:1.8) () {remove $m_2$};
        \end{tikzpicture}
        \caption{An example with $d = 2$, $n = 2$, $r = 1$, $k = 1$. Two rows of the input matrix $M$ are represented as two points $m_1$ and $m_2$ on the plane. The same plane represents the choice of $1$-dimensional approximation subspace through the selection of a vector $V$ orthogonal to it. The algebraic set $W$ is the unit circle since length of $V$ must be 1. The diagonal lines mark the values of $V$ for which $m_1$ and $m_2$ are equidistant. They split $W$ into four one-dimensional and four zero-dimensional cells. For each of the one-dimensional cells it is shown in the corresponding sector which of the points is the outlier and hence is removed.}
    \label{fig:solspace}
\end{figure}
    
    Putting all together, our algorithm proceeds as follows.
    \begin{enumerate}
        \item Using the algorithm from Proposition~\ref{prop:sample_points}, obtain a point $V_C$ from each cell $C$ of $W$ over $\mathcal{P}$.
        \item For each $V_C$, compute the optimal $S_{V_C}$: select the $k$ rows of $M$ with the largest value of $||V_C m_i^T||_F^2$. Construct the instance $(M - S_{V_C}, r)$ of PCA.
        \item The solution to the original instance of \probPCAOut is the best solution among the solutions of all the constructed PCA instances.
    \end{enumerate}

    Since degrees of $Q$ and polynomials from $\mathcal{P}$ are at most 4, $|\mathcal{P}| = \binom{n}{2}$, and the real dimension of $W$ is at most
    $(d - r)d$,
    which is the dimension of $\mathbb{R}^{(d - r) \times d} \supset W$, the algorithm from Proposition~\ref{prop:sample_points} does at most
    $$t=\binom{n}{2}^{(d - r)d} 2^{\Oh(d)}$$
    operations and produces at most $t$ points $V_C$, and our algorithm produces one instance of PCA for each computed point.\footnote{As $W$ is restricted by $Q(V) = 0$, its dimension is actually smaller. It could be bounded more precisely as $(d - r)(d + r - 1)/2$, but we omit the calculation in order not to unnecessarily complicate the text.}

    We are also able to obtain a reduction to 
    $$\binom{n}{2}^{rd} 2^{\Oh(d)}$$
    instances of PCA by proceeding in the same manner for slightly different characterization of $r$-dimensional subspaces. Intuitively, now points on the algebraic set define the orthonormal basis of the subspace itself, and not of its orthogonal complement as in the previous part.

    Now the matrix space is $\mathbb{R}^{r \times d}$, the conditions that an element $V \in \mathbb{R}^{r \times d}$ defines an orthonormal basis of size $r$ are analogous:
    \begin{gather*}
        \bar{Q}^O_{i,j}(V) := \sum_{l = 1}^d v_{il} v_{jl} = 0, \text{ for }1 \le i < j \le r,\\
        \bar{Q}^N_{i}(V) := \left(\sum_{l = 1}^d v_{il}^2\right) - 1 = 0, \text{ for } 1 \le i \le r.
    \end{gather*}
    Again, we may write them as a single polynomial condition $\bar{Q}(V) = 0$ where
    $$\bar{Q}(V) = \sum_{1 \le i < j \le r} (\bar{Q}^O_{i, j}(V))^2 + \sum_{i = 1}^{r} (\bar{Q}^N_i(V))^2,\\$$

    Consider an algebraic set $\bar{W} \subset \mathbb{R}^{r \times d}$ defined as the set of zeroes of $\bar{Q}(V)$. Similarly, any $V \in \bar{W}$ defines an $r$-dimensional subspace $U \in \mathbb{R}^d$ which is the span of the rows of $V$. Since the rows of $V$ form an orthonormal basis of $U$, for any point $x \in \mathbb{R}^d$ the squared distance from $x$ to $U$ is equal to
    $$||x||_F^2 - \sum_{i = 1}^{r} (v_i \cdot x)^2 = ||x||_F^2 - ||V x^T||_F^2.$$

    The new distance formula leads to a slightly different set of polynomials $\bar{\mathcal{P}} = \{\bar{P_{i, j}}\}_{1 \le i < j \le n}$ on $W$, comparing the distance from $m_i$ and from $m_j$,
    $$\bar{P_{i, j}}(V) = (||m_i||_F^2 - ||V m_i^T||_F^2) - (||m_j||_F^2 - ||V m_j^T||_F^2).$$
    Again, the $k$ farthest points and the matrix $S_V$ are the same over any cell in the partition of $\bar{W}$ over $\bar{\mathcal{P}}$. So by the same reasoning as in the first part, it suffices to take a point $V$ from each cell, compute the outlier matrix $S_V$ and solve PCA for $(M - S_V, r)$.

    As we can choose the most effecient of the two subspace representations, we can reduce \probPCAOut to 
    $$\binom{n}{2}^{\min(rd, (d - r)d)} 2^{\Oh(d)}$$
    instances of PCA.\footnote{As with $W$, the dimension of $\bar{W}$ could be bounded more precisely as $r(2d - r - 1)/2$, and with these dimension bounds \probPCAOut reduces to $\binom{n}{2}^{\min(r(2d - r - 1)/2, (d - r)(d + r - 1)/2)} 2^{\Oh(d)}$ instances of PCA.}
\end{proof}

\iffull

\section{ETH lower bound}\label{sec:hard}
In this section we show that  we cannot avoid the dependence on $d$ in the exponent of the running time of a constant factor approximation algorithm for  \probPCAOut unless  
\emph{Exponential Time Hypothesis} (ETH) is false. 
Recall that ETH is the conjecture stated by Impagliazzo, Paturi and Zane~\cite{ImpagliazzoPZ01} in 2001 that for every integer $k\geq 3$, there is a positive constant $\delta_k$ such that the \textsc{$k$-Satisfiability} problem with $n$ variables and $m$ clauses cannot be solved in time $\Oh(2^{\delta_kn}\cdot(n+m)^{\Oh(1)})$. This means that 
\textsc{$k$-Satisfiability} cannot be solved
in subexponential in $n$ time.  

%

We need some additional notation and auxiliary algebraic results about rank of a perturbed matrix.
For a positive integer $r$,  we denote by $I_r$ the $r\times r$ identity matrix. 
For a $n\times m$ matrix $A=(a_{ij})$ and sets of indices $\mathcal{I}\subseteq \{1,\ldots,n\}$ and $\mathcal{J}\subseteq \{1,\ldots,m\}$, we use $A[\mathcal{I},\mathcal{J}]$ to denote the submatrix $A'=(a_{ij})_{i\in\mathcal{I},j\in\mathcal{J}}$ and we say that $A'$ is the \emph{$\mathcal{I}\times\mathcal{J}$-submatrix} of $A$. By $\sigma_{min}(A)$ we denote the smallest singular value of a matrix $A$.

Let $(M,r,k)$ be an instance of \probPCAOut where $M$ is an $n\times d$ matrix. 
We say that a pair of $n\times d$ matrices $(L,S)$ is a \emph{feasible solution} for   $(M,r,k)$  if $\rank(L)\leq r$ and $S$ has at most $k$ non-zero rows. A feasible solution $(L^*,S^*)$ is \emph{optimal} if $\|M-L^*-S^*\|_F^2$ is minimum, and we denote ${\sf Opt}(M,r,k)=\|M-L^*-S^*\|_F^2$. 

The following lemma is a well-known bound that follows from the early results of Weyl~\cite{Weyl1912}. 

\begin{lemma}[\cite{Weyl1912}]\label{lem:weyl}
Let $A$ and $\Delta$ be $r\times r$-matrices and assume that $A$ has full rank. Then if $\|\Delta\|_F<\sigma_{min}(A)$, then $\rank (A+\Delta)=r$. 
\end{lemma}

Lemma~\ref{lem:weyl} implies the next lemma.

\begin{lemma}\label{lem:diag}
Let $A=(a_{ij})$ be a diagonal $r\times r$ matrix with non-zero diagonal elements and let $B=(b_{ij})$ be an $r\times r$ matrix such that for every $i,j\in\{1,\ldots,r\}$, $|b_{ij}|/|a_{ii}|< 1/r$. 
Then $\rank(A+B)=r$.  
\end{lemma}

\begin{proof}
Observe that $\rank(A+B)=\rank(I_r+\Delta)$ where $\Delta=(\delta_{ij})$ is an $r\times r$ matrix with $\delta_{ij}=b_{ij}/a_{ii}$ for $i,j\in\{1,\ldots,r\}$. Then
\[\|\Delta\|_F=\sqrt{\sum_{i,j=1}^r \delta_{ij}^2}= \sqrt{\sum_{i,j=1}^r \Big(\frac{b_{ij}}{a_{ii}}\Big)^2}<\sqrt{\sum_{i,j=1}^r \frac{1}{r^2}} \leq 1=\sigma_{min}(I_r)\]
and $\rank(A+B)=\rank(I_r+\Delta)=r$  by Lemma~\ref{lem:weyl}.
\end{proof}

We also need the following folklore observation about the norm of the inverse of a perturbed matrix; we provide a proof for completeness.

\begin{lemma}\label{lem:perturb}
Let $\Delta$ be an $r\times r$-matrix with $\|\Delta\|_F\leq 1/2$. Then $\rank(I_r+\Delta)=r$ and 
\begin{equation}\label{eq:perturb}
\|(I_r+\Delta)^{-1}\|_F\leq  \sqrt{r}+1
\end{equation}
\end{lemma}

\begin{proof}
Since $\sigma_{min}(I_r)=1>\|\Delta\|_F$, by Lemma~\ref{lem:weyl}, $\rank(I_r+\Delta)=r$. Then 
$$(I_r+\Delta)^{-1}=I_r-\Delta+\Delta^2-\Delta^3+\cdots$$ 
and we have that
$$\|(I_r+\Delta)^{-1}\|_F\leq \|I_r\|_F+\sum_{i=1}^{+\infty}\|\Delta\|_F^i= \sqrt{r}-1+\frac{1}{1-\|\Delta\|_F}\leq \sqrt{r}+1,$$
where the last inequality holds because $\|\Delta\|_F<1/2$.
\end{proof}

Now we are ready to prove the main theorem. 

\begin{theorem}\label{thm:hard-dim}
For any $\omega \geq 1$, there is no $\omega$-approximation algorithm for \probPCAOut with running time $f(d)\cdot N^{o(d)}$ for any computable function $f$ unless ETH fails, where $N$ is the bitsize of the input matrix $M$.
\end{theorem}

\begin{proof}
Let $\omega \geq 1$. 
We reduce from the \textsc{Multicolored Clique} problem:

\defparproblema{\textsc{Multicolored Clique}}{A graph $G$ with a partition $V_1,\ldots,V_r$ of the vertex set.}{$r$}{Decide whether there is a clique $X$ in $G$ with $|X\cap V_i|=1$.}

\noindent
This problem cannot be solved in time $f(r)\cdot |V(G)|^{o(r)}$ for any computable function $f$ unless ETH fails~\cite{CyganFKLMPPS15,LokshtanovMS11}).
 
Let $(G,V_1,\ldots,V_r)$ be an instance of \textsc{Multicolored Clique}. We assume without loss of generality that for $i\in\{1,\ldots,r\}$, $|V_i|=n$  (otherwise, we  can add dummy isolated vertices to insure the property; clearly, the claim that the problem cannot be solved in time $f(r)\cdot |V(G)|^{o(r)}$ up to ETH remains correct) and each $V_i$ is an independent set.  
We denote $v_1^i,\ldots,v_n^i$ the vertices of $V_i$ for $i\in\{1,\ldots,r\}$ and let $m=|E(G)|$. We also assume that $r\geq 4$.

We set $a=4(r+1)^2mn^2\omega$ and $c=9a$.
We define $m\times r$ matrices $P=(p_{ij})$ and $Q=(q_{ij})$ whose rows are indexed by the edges of $G$ as follows: for every $e=v_i^sv_j^t\in E(G)$, 
\begin{itemize}
\item set  $p_{e s}=c(a-i)$, $p_{e t}=c(a-j)$,
\item set $q_{e s}=q_{e t}=ca$,
\item for $h\in\{1,\ldots,r\}$ such that $h\neq i,j$, set $p_{e i}=q_{e j}=0$.
\end{itemize}
For $e\in E(G)$, $p_e$ and $q_e$ denotes the $e$-th row of $P$ and $Q$, respectively.
We define the $r\times r$ matrix $A=aI_r$.
Let $a_1,\ldots,a_r$ be the rows of $A$.
We construct $2m$ copies of $A$ denoted by $A_i,A'_i$ for $i\in\{1,\ldots,m\}$.
We construct an $(r+1)m\times 2r$ matrix $M$ using $P$, $Q$ and $A_i,A_i'$ for $i\in\{1,\ldots,m\}$ as blocks:
\begin{equation}\label{eq:M}
M=
\left(
\begin{array}{c|c}
A_1 & A_1'\\
\hline
\vdots & \vdots \\
\hline
A_m & A_m'\\
\hline
P & Q
\end{array}
\right).
\end{equation}
To simplify notation, we index the rows of $M$ corresponding to $(P\mid Q)$ by the edges of $G$ and use integers for other indices. We follow the same convention of the matrices of similar structure that are considered further. 
We define $k=m-\binom{r}{2}$ and 
set $D=rmn^2$ and $D'=D\omega$.  
Note that the dimension $d=2r$
Observe also that $a$ and $c$ are chosen in such a way that $c>>a>>D'$ and this is crucial for our reduction.
The main property of the constructed instance of \probPCAOut is stated in the following claim.

%

\begin{claim}\label{cl:crucial}
If $(G,V_1,\ldots,V_r)$ is a yes-instance of \textsc{Multicolored Clique}, then ${\sf Opt}(M,r,k)\leq D$, and if there is a feasible solution $(L,S)$ for $(M,r,k)$ with 
$\|M-L-S\|_F^2\leq \omega D$, then $(G,V_1,\ldots,V_r)$ is a yes-instance of \textsc{Multicolored Clique}. 
\end{claim}

Suppose that $(G,V_1,\ldots,V_r)$ is a yes-instance of \textsc{Multicolored Clique}, that is, there is a clique $X$ of $G$ such that $|X\cap V_i|=1$ for $i\in\{1,\ldots,r\}$. Let 
$X=\{v_{i_1}^1,\ldots,v_{i_r}^r\}$ and $R=\{v_{i_s}^sv_{i_t}^t\mid 1\leq s<t\leq r\}$, that is, $R=E(G[X])$. We show that ${\sf Opt}(M,r,k)\leq D$. For this, we construct a feasible solution $(L,S)$ such that $\|M-L-S\|_F^2\leq D$.

We define the $m\times r$ matrices $P^*$ and $Q^*$ by setting the elements of $P$ and $Q$ respectively that are in the rows for $e\in E(G)\setminus R$ to be zero and the other elements are the same as in $P$ and $Q$, that is, for $e\in R$, the rows of $P^*$ and $Q^*$ are $p_e$ and $q_e$ respectively and other rows are zero-rows. 
We define an $r\times r$ matrix $B$ as follows:
\[
B=
\begin{pmatrix}
a-i_1& 0 & \ldots& 0\\
0 & a-i_2& \ldots& 0\\
\vdots & \vdots & \ddots & \vdots\\ 
0 &0 &\ldots & a-i_r
\end{pmatrix},
\]
that is, the diagonal elements of $B$ are $a-i_1,\ldots,a-i_r$ and the other elements are zeros. 
Denote by $b_1,\ldots,b_r$ the rows of $B$.
We construct $m$ copies $B_1,\ldots,B_m$  of $B$ and 
 the $(r+1)m\times 2k$ matrix $L$ using $P^*$, $Q^*$ and $B_i,A_i'$ for $i\in\{1,\ldots,m\}$ as blocks:
\begin{equation}\label{eq:S}
L=
\left(
\begin{array}{c|c}
B_1 & A_1'\\
\hline
\vdots & \vdots \\
\hline
B_m & A_m'\\
\hline
P^* & Q^*
\end{array}
\right).
\end{equation}
It is straightforward to verify that $\rank(L)\leq r$. Indeed, the rank of the each submatrix $(B_i\mid A_i')$ is $r$ as only diagonal elements of $B_i$ and $A_i'$ are non-zero. Also for each $e=v_{i_s}^sv_{i_t}^t\in R$, we have that $p_e=c(b_{i_s}+b_{i_t})$ and $q_e=c(a_{i_s}+a_{i_t})$, i.e., each row of $L$ indexed by $e\in R$ is a linear combination of two rows of  $(B_1\mid A_1')$.

Let $\mathbb{0}$ be $r\times r$ zero matrix. We construct $2m$ copies $\mathbb{0}_1,\ldots,\mathbb{0}_m$ and $\mathbb{0}_1',\ldots,\mathbb{0}_m'$ of $\mathbb{0}$ and define
\begin{equation}\label{eq:L}
S=
\left(
\begin{array}{c|c}
\mathbb{0}_1 & \mathbb{0}_1'\\
\hline
\vdots & \vdots \\
\hline
\mathbb{0}_m & \mathbb{0}_m'\\
\hline
P-P^* & Q-Q^*
\end{array}
\right).
\end{equation}
Clearly, this matrix has at most $k$ non-zero rows that are indexed by $e\in R$.

Combining (\ref{eq:M})--(\ref{eq:L}), we have that
\begin{align*}
M-S-L&\\
=&\left(
\begin{array}{c|c}
A_1 & A_1'\\
\hline
\vdots & \vdots \\
\hline
A_m & A_m'\\
\hline
P & Q
\end{array}
\right)
-
\left(
\begin{array}{c|c}
B_1 & A_1'\\
\hline
\vdots & \vdots \\
\hline
B_m & A_m'\\
\hline
P^* & Q^*
\end{array}
\right)
-
\left(
\begin{array}{c|c}
\mathbb{0}_1 & \mathbb{0}_1'\\
\hline
\vdots & \vdots \\
\hline
\mathbb{0}_m & \mathbb{0}_m'\\
\hline
P-P^* & Q-Q^*
\end{array}
\right)\\
=&
\left(
\begin{array}{c|c}
A_1-B_1 & \mathbb{0}\\
\hline
\vdots & \vdots \\
\hline
A_m-B_m & \mathbb{0}\\
\hline
\mathbb{0}' & \mathbb{0}'
\end{array}
\right),
\end{align*}
where $\mathbb{0}'$ is the $m\times r$ zero matrix,
and
\[\|M-S-L\|_F^2=m\|A-B\|_F^2=m(i_1^2+\ldots+i_r^2)\leq rmn^2=D.\]

We conclude that $(L,S)$ is a feasible solution for the considered instance $(M,r,k)$ of  \probPCAOut with $\|M-S-L\|_F^2\leq D$. Therefore, ${\sf Opt}(M,r,k)\leq D$.

\medskip
Suppose now that $(L,S)$ is a feasible solution for  $(M,r,k)$ of  \probPCAOut with $\|M-S-L\|_F^2\leq \omega D=D'$. We prove that $(G,V_1,\ldots,V_r)$ is a yes-instance of \textsc{Multicolored Clique}. 

 Recall that $S$ has at most $k=m-\binom{r}{2}$ non-zero rows, Hence, there is a set $R\subseteq E(G)$ with $|R|=m-k=\binom{r}{2}$ such that the rows of $S$ indexed by $e\in R$ are zero-rows. We claim that the edges of $R$ form the set of edges of a complete graph. More formally, we show the following.

\begin{claim}\label{cl:clique}
There are $i_1,\ldots,i_r\in\{1,\ldots,n\}$ such that $R=\{v_{i_s}^sv_{i_t}^t\mid 1\leq s<t\leq r\}$.
\end{claim}

To prove Claim~\ref{cl:clique}, we need some auxiliary results. 
Let 
$$
L=\left(
\begin{array}{c|c}
L_1 & L_1'\\
\hline
\vdots & \vdots \\
\hline
L_m & L_m'\\
\hline
L_{m+1} & L_{m+1}'  
\end{array}
\right)
\text{ and }
S=\left(
\begin{array}{c|c}
S_1 & S_1'\\
\hline
\vdots & \vdots \\
\hline
S_m & S_m'\\
\hline
S_{m+1} & S_{m+1}'  
\end{array}
\right), 
$$
where each $L_i,L_i',S_i,S_i'$ is an $r\times r$ submatrix of $L$ and $S$ respectively for $i\in\{1,\ldots,m\}$. 
Since $S$ has at most $k$ non-zero rows and $k<m$, there is $i\in\{1,\ldots,m\}$ such that $(S_i,S_i')=(\mathbb{0}\mid \mathbb{0})$. We assume without loss of generality that $i=1$. 
Let $P_R$ and $Q_R$ be the $R\times \{1,\ldots,r\}$-submatrices of $P$ and $Q$ respectively, and denote by $X$ and $Y$ the $R\times \{1,\ldots,r\}$-submatrices of $L_{m+1}$ and $L_{m+1}'$ respectively. Let
$$
M^*=\left(
\begin{array}{c|c}
A_1 & A_1'\\
\hline
P_R& Q_R  
\end{array}
\right) 
\text{ and }
L^*=\left(
\begin{array}{c|c}
L_1 & L_1'\\
\hline
X & X' 
\end{array}
\right).
$$
We have that $(i)$ $\rank(L^*)\leq\rank(L)\leq r$ and $(ii)$ $\|M^*-L^*\|_F^2\leq \|M-L-S\|_F^2\leq D'$.
Recall that $A_1=A_1'=A=aI_r$. Therefore, we have that $\rank(A_1')=r$ and $\sigma_{min}(A_1')=a>\sqrt{D'}$. Since 
$A_1'-L_1'$ is a submatrix of  $M^*-L^*$ and by statement $(ii)$, we have that $\|A_1-L_1'\|_F\leq \|M^*-L^*\|_F\leq \sqrt{D'}<\sigma_{min}(A_1')$. 
Then, by Lemma~\ref{lem:weyl} (substitute $A=A_1'$ and $\Delta=A_1'-L_1'$ in Lemma~\ref{lem:weyl}), we conclude that $\rank(L_1')=\rank(-L_1')=\rank(A_1')=r$. Therefore, by  $(i)$,  $\rank(L^*)=r$ and the first $r$ rows of $L^*$ are linearly independent and form a row basis of $L^*$.  In particular, the rows of $L^*$ indexed by elements in $R$ are linear combinations of the rows of this basis.

It is convenient for us to switch from the basis formed by the rows of $(L_1\mid L_1')$ to a more specific basis. Since $\rank(L_1')=r=\rank(A)$, there is a unique $r\times r$ matrix $\Lambda$ such that $\Lambda L_1'=A$. This implies that $\Lambda (L_1\mid L_1')=(Y\mid A)$ for some $r\times r$ matrix $Y$. In other words, the rows of $Z=(Y\mid A)$ are linear combinations of the rows of  $(L_1\mid L_1')$.  Since $\rank(A)=r$, 
we have that the rows $L^*$ are linear combinations of the rows of $Z$.
Then for 
\[\hat{L}=
\left(
\begin{array}{c|c}
Y & A\\
\hline
X & X' 
\end{array}
\right),
\]
it holds that $\rank(\hat{L})=r$.

By the definition of $\Lambda$, we have that $\Lambda=A(L_1')^{-1}=aI_r(L_1')^{-1}=a(L_1')^{-1}$. Notice that $a(L_1')^{-1}=(\frac{1}{a} L_1')^{-1}=(I_r+\frac{1}{a} L_1'-I_r)^{-1}$. 
Let $\Theta=\frac{1}{a}(L_1'-A_1)$. Since $\|A_1-L_1'\|_F\leq \|M^*-L^*\|_F\leq \sqrt{D'}$, $\|\Theta\|_F\leq \frac{1}{a}\sqrt{D'}\leq 1/2$.  
We have that $\Lambda=AL_1'^{-1}=(I_r+\Theta)^{-1}$ and $\|\Lambda\|\leq \sqrt{r}+1$ by Lemma~\ref{lem:perturb}.
It holds that 
\[Y=\Lambda L_1=\Lambda(L_1'+(L_1-L_1'))=\Lambda L_1'+\Lambda(L_1-L_1')=A+\Lambda(L_1-L_1'),\]
and for $\Delta=\Lambda(L_1-L_1')$,
\[\|\Delta\|_F\leq \|\Lambda\|\|L_1-L_1'\|_F=\|\Lambda\|(\|L_1-A\|_F+\|L_1'-A\|_F)\leq 2(\sqrt{r}+1)\sqrt{D'}.\]

We  obtain that for every $e\in E(G)$, the rows of $L^*$ indexed by $e$ is a linear combination of the rows of $Z=(A+\Delta\mid A)$ where $\|\Delta\|_F\leq 2(\sqrt{r}+1)\sqrt{D'}$.
This property is crucial for the proof of Claim~\ref{cl:clique}.

Let $\Delta=(\delta_{ij})$. Note that we have that for every $i,j\in\{1,\ldots,r\}$, $|\delta_{ij}|\leq 2(\sqrt{r}+1)\sqrt{D'}$.
Denote by $\Xi=X-P_R$ and $\Xi'=X'-Q_R$. Let $\Xi=(\xi_{ej})$ and $\Xi'=(\xi_{ej}')$. Since $\|P_R-X\|_F\leq \|M^*-L^*\|\leq \sqrt{D'}$ and
$\|Q_R-Y\|_F\leq \|M^*-L^*\|\leq \sqrt{D'}$, $|\xi_{ej}|\leq \sqrt{D'}$ and $|\xi_{ej}'|\leq\sqrt{D'}$ 
for $e\in R $ and $j\in\{1,\ldots,r\}$.

Observe that now we can write that 
\[\hat{L}=
\left(
\begin{array}{c|c}
A+\Delta & A\\
\hline
P_R+\Xi & Q_R+\Xi' 
\end{array}
\right).
\]

To prove Claim~\ref{cl:clique}, we first show the following claim.

\begin{claim}\label{cl:first}
If $v_s^iv_t^j,v_{s'}^{i'}v_{t'}^{j'}\in R$ for some $i,j,i',j'\in\{1,\ldots,r\}$ and $s,s',t,t'\in\{1,\ldots,n\}$, then $i\neq i'$ or $j\neq j'$.  
\end{claim}

\begin{proof}[Proof of Claim~\ref{cl:first}]
To obtain a contradiction, assume that there are $i,j\in\{1,\ldots,r\}$ and $s,s',t,t'\in\{1,\ldots,n\}$ such that $e=v_s^iv_t^j$ and $e'=v_{s'}^iv_{t'}^j$ are distinct edges of $R$. 
As either $s\neq s'$ or $t\neq t'$, we can assume without loss of generality that $s\neq s'$ using symmetry. Let 
$\mathcal{I}=\{1,\ldots,r\}\cup\{e,e'\}$ and $\mathcal{J}=\{i\}\cup\{1+r,\ldots,2r\}$. We show that $\rank(\hat{L}[\mathcal{I},\mathcal{J}])\geq r+1$. We can write $\hat{L}[\mathcal{I},\mathcal{J}]$ as follows assuming that $i<j$ (the case $i>j$ is symmetric):
{\small
\[
\left(
\begin{array}{c|ccc|c|ccc|c|ccc}
\delta_{1i}   & a &\cdots &0&0&0&\cdots  &0    &0       &0&\cdots&0\\
\vdots          &    &\ddots&  & \vdots  &  &\ddots &     &\vdots&  &\ddots&   \\
\delta_{i-1i}& 0 &  \cdots&a&0&0&\cdots&0&0&0&\cdots&0\\ 
\hline
a+\delta_{ii}& 0 & \cdots&0&a&0&\cdots&0&0&0&\cdots&0\\ 
\hline
\delta_{i+1i}& 0 &  \cdots&0&0&a&\cdots&0&0&0&\cdots&0\\ 
\vdots          &    & \ddots&  & \vdots &  &\ddots &     &\vdots&  &\ddots&  \\
\delta_{j-1i}& 0 &  \cdots&0&0&0&\cdots&a&0&0&\cdots&0\\ 
\hline
\delta_{ji}& 0 &  \cdots&0&0&0&\cdots&0&a&0&\cdots&0\\ 
\hline
\delta_{j+1i}& 0 & \cdots&0&0&0&\cdots&0&0&a&\cdots&0\\ 
\vdots          &    &  \ddots&  & \vdots  &  &\ddots &     &\vdots&  &\ddots&  \\
\delta_{ri}& 0 & \cdots&0&0&0&\cdots&0&0&0&\cdots&a\\ 
\hline
c(a-s)+\xi_{ei}& \xi_{e1}' & \cdots& \xi_{ei-1}'&ca+\xi_{ei}'&\xi_{ei+1}'&\cdots&\xi_{ej-1}'&ca+\xi_{ej}'&\xi_{ej+1}'&\cdots&\xi_{er}'\\
c(a-s')+\xi_{e'i}& \xi_{e'1}' & \cdots& \xi_{e'i-1}'&ca+\xi_{e'i}'&\xi_{e'i+1}'&\cdots&\xi_{e'j-1}'&ca+\xi_{e'j}'&\xi_{e'j+1}'&\cdots&\xi_{e'r}'\\
\end{array}
\right).
\]}
We subtract  the $i$-th and $j$-th rows multiplied by $c$ from the last two rows and obtain the matrix:
{\small
 \[
\left(
\begin{array}{c|ccc|c|ccc|c|ccc}
\delta_{1i}   & a &\cdots &0&0&0&\cdots  &0    &0       &0&\cdots&0\\
\vdots          &    & \ddots&  & \vdots  &  &\ddots &     &\vdots&  &\ddots&   \\
\delta_{i-1i}& 0 &  \cdots&a&0&0&\cdots&0&0&0&\cdots&0\\ 
\hline
a+\delta_{ii}& 0 &  \cdots&0&a&0&\cdots&0&0&0&\cdots&0\\ 
\hline
\delta_{i+1i}& 0 &  \cdots&0&0&a&\cdots&0&0&0&\cdots&0\\ 
\vdots          &    & \ddots&  & \vdots  &  &\ddots &     &\vdots&  &\ddots&  \\
\delta_{j-1i}& 0 & \cdots&0&0&0&\cdots&a&0&0&\cdots&0\\ 
\hline
\delta_{ji}& 0 &  \cdots&0&0&0&\cdots&0&a&0&\cdots&0\\ 
\hline
\delta_{j+1i}& 0 &  \cdots&0&0&0&\cdots&0&0&a&\cdots&0\\ 
\vdots          &    & \ddots&  & \vdots  &  &\ddots &     &\vdots&  &\ddots&  \\
\delta_{ri}& 0 &  \cdots&0&0&0&\cdots&0&0&0&\cdots&a\\ 
\hline
-cs+\xi_{ei}-c(\delta_{ii}+\delta_{ji})& \xi_{e1}' & \cdots& \xi_{ei-1}'&\xi_{ei}'&\xi_{ei+1}'&\cdots&\xi_{ej-1}'&\xi_{ej}'&\xi_{ej+1}'&\cdots&\xi_{er}'\\
-cs'+\xi_{e'i}-c(\delta_{ii}+\delta_{ji})& \xi_{e'1}' & \cdots& \xi_{e'i-1}'&\xi_{e'i}'&\xi_{e'i+1}'&\cdots&\xi_{e'j-1}'&\xi_{e'j}'&\xi_{e'j+1}'&\cdots&\xi_{e'r}'\\
\end{array}
\right).
\]}
Then we subtract the last row from the previous and delete the last row:
{\small
 \[
\left(
\begin{array}{c|ccc|c|ccc}
\delta_{1i}   & a &\cdots &0&0&0&\cdots  &0 \\
\vdots          &    & \ddots&  &  \vdots &  &\ddots   \\
\delta_{i-1i}& 0 &  \cdots&a&0&0&\cdots&0\\ 
\hline
a+\delta_{ii}& 0 &  \cdots&0&a&0&\cdots&0\\ 
\hline
\delta_{i+1i}& 0 &  \cdots&0&0&a&\cdots&0\\ 
\vdots         &   & \ddots &  & \vdots  &  &\ddots &  \\
\delta_{ri}& 0 &  \cdots&0&0&0&\cdots&a\\
\hline
c(s'-s)+\xi_{ei}-\xi_{e'i}& \xi_{e1}'-\xi_{e'1}' & \cdots&\xi_{ei-1}'-\xi_{e'i-1}' &\xi_{ei}'-\xi_{e'i}' &\xi_{e i+1}'-\xi_{e'i+1}' &\cdots& \xi_{er}'-\xi_{e'r}'
\end{array}
\right).
\]}
Note that $|c(s'-s)+\xi_{ei}-\xi_{e'i}|\geq |c(s'-s)|-2\sqrt{D'}>0$. Let $\beta=a/(c(s'-s)+\xi_{ei}-\xi_{e'i})$. Let also $\alpha_h=\xi_{eh}'-\xi_{e'h}'$ and $\beta_h=-\beta(\xi_{eh}'-\xi_{e'h}')$ for $h\in\{1,\ldots,r\}$. 
We subtract from the $i$-th row the last row multiplied by $\beta$ and made the last row the first one. We obtain the following matrix:
{\small
 \[U=
\left(
\begin{array}{c|ccc|c|ccc}
c(s'-s)+\xi_{ei}-\xi_{e'i}&\alpha_1 & \cdots&\alpha_{i-1} &\alpha_i &\alpha_{i+1} &\cdots& \alpha_r\\
\hline
\delta_{1i}   & a &\cdots &0&0&0&\cdots  &0 \\
\vdots          &    & \ddots&  & \vdots  &  &\ddots   \\
\delta_{i-1i}& 0 &  \cdots&a&0&0&\cdots&0\\ 
\hline
\delta_{ii}& \beta_1 &  \cdots&\beta_{i-1}&a+\beta_i&\beta_{i+1}&\cdots&\beta_r\\ 
\hline
\delta_{i+1i}& 0 &  \cdots&0&0&a&\cdots&0\\ 
\vdots         &   & \ddots &  & \vdots  &  &\ddots &  \\
\delta_{ri}& 0 &  \cdots&0&0&0&\cdots&a
\end{array}
\right).
\]}

Let 
{\small
\[
 W=
\left(
\begin{array}{c|ccc}
c(s'-s) & 0& \cdots&0 \\
\hline
0   & a &\cdots &0\\
\vdots &    & \ddots&   \\
0  &0&\cdots&a
\end{array}
\right). 
\]}
and
 {\small
 \[ W\rq{}=
\left(
\begin{array}{c|ccc|c|ccc}
\xi_{ei}-\xi_{e'i}&\alpha_1 & \cdots&\alpha_{i-1} &\alpha_i &\alpha_{i+1} &\cdots& \alpha_r\\
\hline
\delta_{1i}   & 0 &\cdots &0&0&0&\cdots  &0 \\
\vdots          &    & \ddots&  & \vdots  &  &\ddots   \\
\delta_{i-1i}& 0 &  \cdots&0&0&0&\cdots&0\\ 
\hline
\delta_{ii}& \beta_1 &  \cdots&\beta_{i-1}&\beta_i&\beta_{i+1}&\cdots&\beta_r\\ 
\hline
\delta_{i+1i}& 0 &  \cdots&0&0&0&\cdots&0\\ 
\vdots         &   & \ddots &  & \vdots  &  &\ddots &  \\
\delta_{ri}& 0 &  \cdots&0&0&0&\cdots&0
\end{array}
\right) 
\]}

Notice that $U=W+W\rq{}$, where $W$ is an $(r+1)\times (r+1)$ diagonal matrix.  
Since $|c(s'-s)|\geq c>a$, we have that  the absolute value of each of the diagonal elements of $W$ is at least $a$. 
Recall that   $|\delta_{pq}|\leq 2(\sqrt{r}+1)\sqrt{D'}$, $|\xi_{pq}|\leq \sqrt{D'}$ and $|\xi_{pq}'|\leq\sqrt{D'}$ for $p,q\in\{1,\ldots,r\}$.
We have that  $0<|\beta|<1$, $|\alpha_h|\leq 2\sqrt{D'}$ and $|\beta_h|\leq  2\sqrt{D'}$ for $h\in\{1,\ldots,r\}$. 
This implies that the absolute value of each of the elements of $W\rq{}$ is at most $2(\sqrt{r}+1) \sqrt{D'}$.
Hence, by Lemma~\ref{lem:diag}, $\rank(U)=\rank(W+W\rq{})=\rank(W)=r+1$. 
%
That is, $\rank(U)=r+1\leq \rank (\hat{L}[\mathcal{I},\mathcal{J}])\leq \rank(\hat{L})=r$; a contradiction.
\end{proof}

By Claim~\ref{cl:first}, $R$ has no two edges with their end-vertices in the same sets of the partition $V_1,V_2,\ldots,V_r$. In particular, this means that for each $i\in\{1,\ldots,r\}$, $R$ contains exactly $r-1$ edges incident to vertices of $V_i$. To prove Claim~\ref{cl:clique}, we will argue that these edges are incident to the same vertex of $V_i$. We need the following auxiliary claim.

\begin{claim}\label{cl:second}
For all distinct $i,j\in\{1,\ldots,r\}$, $|\delta_{ij}|<1/4$.
\end{claim}

\begin{proof}[Proof of Claim~\ref{cl:second}]
We show that for every $i\in\{1,\ldots,r\}$ and all pairs of distinct $p,q\in\{1,\ldots,r\}$ such that $p,q\neq i$, $|\delta_{pi}+\delta_{qi}|<1/8$. 
The proof is similar to the proof of Claim~\ref{cl:first}. To obtain a contradiction, assume that for some $i\in\{1,\ldots,r\}$, there are distinct $p,q\in\{1,\ldots,r\}$ such that $p,q\neq i$ and 
$|\delta_{pi}+\delta_{qi}|\geq1/8$. By Claim~\ref{cl:first}, there is $e=v_s^pv_t^q\in R$ for some $s,t\in \{1,\ldots,n\}$.
Let 
$\mathcal{I}=\{1,\ldots,r\}\cup\{e\}$ and $\mathcal{J}=\{i\}\cup\{1+r,\ldots,2r\}$. We show that $\rank(\hat{L}[\mathcal{I},\mathcal{J}])\geq r+1$.
For this, we write  $\hat{L}[\mathcal{I},\mathcal{J}]$ assuming that $i<p<q$ (the other cases are symmetric):
{\small
\[
\left(
\begin{array}{c|ccc|c|ccc|c|ccc|c|ccc}
\delta_{1i}   & a &\cdots &0&0&0&\cdots  &0    &0       &0&\cdots&0&0&0&\cdots&0\\
\vdots          &    &\ddots&  &  \vdots &  &\ddots &     &\vdots&  &\ddots& &\vdots& &\ddots &  \\
\delta_{i-1i}& 0 &  \cdots&a&0&0&\cdots&0&0&0&\cdots&0&0&0&\cdots&0\\ 
\hline
a+\delta_{ii}& 0 & \cdots&0&a&0&\cdots&0&0&0&\cdots&0&0&0&\cdots&0\\ 
\hline
\delta_{i+1i}   & 0 &\cdots &0&0&a&\cdots  &0    &0       &0&\cdots&0&0&0&\cdots&0\\
\vdots          &    &\ddots&  &  \vdots &  &\ddots &     &\vdots&  &\ddots& &\vdots& &\ddots &  \\
\delta_{p-1i}& 0 &  \cdots&0&0&0&\cdots&a&0&0&\cdots&0&0&0&\cdots&0\\ 
\hline
\delta_{pi}& 0 & \cdots&0&0&0&\cdots&0&a&0&\cdots&0&0&0&\cdots&0\\ 
\hline
\delta_{p+1i}   & 0 &\cdots &0&0&0&\cdots  &0    &0       &a&\cdots&0&0&0&\cdots&0\\
\vdots          &    &\ddots&  &  \vdots &  &\ddots &     &\vdots&  &\ddots& &\vdots& &\ddots &  \\
\delta_{q-1i}& 0 &  \cdots&0&0&0&\cdots&0&0&0&\cdots&a&0&0&\cdots&0\\ 
\hline
\delta_{qi}& 0 & \cdots&0&0&0&\cdots&0&0&0&\cdots&0&a&0&\cdots&0\\ 
\hline
\delta_{q+1i}   & 0 &\cdots &0&0&0&\cdots  &0    &0       &0&\cdots&0&0&a&\cdots&0\\
\vdots          &    &\ddots&  &  \vdots &  &\ddots &     &\vdots&  &\ddots& &\vdots& &\ddots &  \\
\delta_{ri}& 0 &  \cdots&0&0&0&\cdots&0&0&0&\cdots&0&0&0&\cdots&a\\ 
\hline
\xi_{ei}& \xi_{e1}' & \cdots& & & &\cdots &\xi_{ep-1}'&ca+\xi_{ep}'&\xi_{ep+1}'&\cdots&\xi_{eq-1}'&ca+\xi_{eq}'&\xi_{eq+1}'&\cdots&\xi_{er}'
\end{array}
\right).
\]}
We subtract  the $i$-th and $j$-th rows multiplied by $c$ from the last row and obtain the matrix:
{\small
\[
\left(
\begin{array}{c|ccc|c|ccc|c|ccc|c|ccc}
\delta_{1i}   & a &\cdots &0&0&0&\cdots  &0    &0       &0&\cdots&0&0&0&\cdots&0\\
\vdots          &    &\ddots&  &  \vdots &  &\ddots &     &\vdots&  &\ddots& &\vdots& &\ddots &  \\
\delta_{i-1i}& 0 &  \cdots&a&0&0&\cdots&0&0&0&\cdots&0&0&0&\cdots&0\\ 
\hline
a+\delta_{ii}& 0 & \cdots&0&a&0&\cdots&0&0&0&\cdots&0&0&0&\cdots&0\\ 
\hline
\delta_{i+1i}   & 0 &\cdots &0&0&a&\cdots  &0    &0       &0&\cdots&0&0&0&\cdots&0\\
\vdots          &    &\ddots&  &  \vdots &  &\ddots &     &\vdots&  &\ddots& &\vdots& &\ddots &  \\
\delta_{p-1i}& 0 &  \cdots&0&0&0&\cdots&a&0&0&\cdots&0&0&0&\cdots&0\\ 
\hline
\delta_{pi}& 0 & \cdots&0&0&0&\cdots&0&a&0&\cdots&0&0&0&\cdots&0\\ 
\hline
\delta_{p+1i}   & 0 &\cdots &0&0&0&\cdots  &0    &0       &a&\cdots&0&0&0&\cdots&0\\
\vdots          &    &\ddots&  &  \vdots &  &\ddots &     &\vdots&  &\ddots& &\vdots& &\ddots &  \\
\delta_{q-1i}& 0 &  \cdots&0&0&0&\cdots&0&0&0&\cdots&a&0&0&\cdots&0\\ 
\hline
\delta_{qi}& 0 & \cdots&0&0&0&\cdots&0&0&0&\cdots&0&a&0&\cdots&0\\ 
\hline
\delta_{q+1i}   & 0 &\cdots &0&0&0&\cdots  &0    &0       &0&\cdots&0&0&a&\cdots&0\\
\vdots          &    &\ddots&  &  \vdots &  &\ddots &     &\vdots&  &\ddots& &\vdots& &\ddots &  \\
\delta_{ri}& 0 &  \cdots&0&0&0&\cdots&0&0&0&\cdots&0&0&0&\cdots&a\\ 
\hline
-c(\delta_{pi}+\delta_{qi})+\xi_{ei}& \xi_{e1}' & \cdots& & & &\cdots &\xi_{ep-1}'&\xi_{ep}'&\xi_{ep+1}'&\cdots&\xi_{eq-1}'&\xi_{eq}'&\xi_{eq+1}'&\cdots&\xi_{er}'
\end{array}
\right).
\]}
 Since $|\delta_{pi}+\delta_{qi}|\geq 1/8$, $|c(\delta_{pi}+\delta_{qi})|>|\xi_{ei}|$. Let $\alpha=a/(c(\delta_{pi}+\delta_{qi})-\xi_{ei})$ and let 
$\alpha_h= \alpha\xi_{eh}'$ for $h\in\{1,\ldots,r\}$. We add to the $i$-th row the last row multiplied by $\alpha$ and then make the last row the first one:
{\small
\[
\left(
\begin{array}{c|ccc|c|ccc}
-c(\delta_{pi}+\delta_{qi})+\xi_{ei}& \xi_{e1}' & \cdots& & & &\cdots &\xi_{er}'\\
\hline
\delta_{1i}   & a &\cdots &0&0&0&\cdots  &0\\   
\vdots          &    &\ddots&  &  \vdots &  &\ddots &       \\
\delta_{i-1i}& 0 &  \cdots&a&0&0&\cdots&0\\ 
\hline
\delta_{ii}& \alpha_1 & \cdots&\alpha_{i-1}&a+\alpha_i&\alpha_{i+1}&\cdots&\alpha_r\\ 
\hline
\delta_{i+1i}   & a &\cdots &0&0&a&\cdots  &0    \\
\vdots          &    &\ddots&  &  \vdots &  &\ddots &       \\
\delta_{ri}& 0 &  \cdots&0&0&0&\cdots& a
\end{array}
\right).
\]}
Denote by $U$ the constructed matrix.
 
Recall that $|\delta_{ji}|\leq 2(\sqrt{r}+1)\sqrt{D'}$  for $j\in\{1,\ldots,r\}$, $|\xi_{ei}|\leq\sqrt{D'}$ and $|\xi_{eh}'|\leq \sqrt{D'}$ for $h\in\{1,\ldots,r\}$.
Since $|\delta_{pi}+\delta_{qi}|\geq 1/8$ and $c=9a$, $|c(\delta_{pi}+\delta_{qi})-\xi_{ei} |>a$ and, therefore, $|\alpha|<1$. Then 
$|\alpha_h|\leq \sqrt{D'}$ for $h\in\{1,\ldots,r\}$.  Then by Lemma~\ref{lem:diag}, $\rank(U)$ is the rank of the $(r+1)\times(r+1)$ diagonal matrix
{
 \[
\left(
\begin{array}{c|ccc}
-c(\delta_{pi}+\delta_{qi})& 0& \cdots&0 \\
\hline
0   & a &\cdots &0\\
\vdots &    & \ddots&   \\
0  &0&\cdots&a
\end{array}
\right),
\]}
that is, $\rank(U)=r+1\leq \rank (\hat{L}[\mathcal{I},\mathcal{J}])\leq \rank(\hat{L})=r$; a contradiction.
This proves that for every $i\in\{1,\ldots,r\}$ and all pairs of distinct $p,q\in\{1,\ldots,r\}$ such that $p,q\neq i$, $|\delta_{pi}+\delta_{qi}|<1/8$. 

To show the statement of Claim~\ref{cl:second}, consider $i\in\{1,\ldots,r\}$. Let $j\in\{1,\ldots,r\}$, $j\neq i$, be such that $|\delta_{ji}|=\max\{\delta_{hi}\mid 1\leq h\leq r,h\neq i\}$. 
Clearly, it is sufficient to prove that $|\delta_{ji}|< 1/4$. 
For the sake of contradiction, assume that  $|\delta_{ji}|\geq 1/4$. 
Because $r\geq 4$, there are distinct $p,q\in\{1,\ldots,r\}\setminus\{i,j\}$. 
We have that $|\delta_{ji}|\geq |\delta_{pi}|$. Hence,  $|\delta_{ji}+\delta_{pi}|=\delta_{ji}+\delta_{pi}$ if $\delta_{ji}>0$ and 
$|\delta_{ji}+\delta_{pi}|=-\delta_{ji}-\delta_{pi}$ otherwise. 
Suppose that $\delta_{ji}>0$.  Then because  $\vert \delta_{ji}+\delta_{pi}\vert<1/8$, we have 
$0\leq \delta_{ji}+\delta_{pi}<1/8$. By the same arguments, $0\leq  \delta_{ji}+\delta_{qi} <1/8$.
Because $\delta_{ji}>1/4$, $-\delta_{pi}>\delta_{ji}-1/8\geq 1/8$ and $-\delta_{qi}>\delta_{ji}-1/8\geq 1/8$. Then 
$|\delta_{pi}+\delta_{qi}|\geq 1/4>1/8$; a contradiction. 
Suppose  $\delta_{ji}<0$. 
Then $|\delta_{ji}+\delta_{pi}|=-\delta_{ji}-\delta_{pi}$ and $|\delta_{ji}+\delta_{qi}|=-\delta_{ji}-\delta_{qi}$ (because of the definition of $j$).
Hence, we have that 
$0\leq -\delta_{ji}-\delta_{pi}<1/8$ and $0\leq -\delta_{ji}-\delta_{qi}<1/8$ and we obtain a contradiction in the same way.
\end{proof}

Now we are ready to make the final step of the proof of Claim~\ref{cl:clique}.

\begin{claim}\label{cl:third}
If $v_s^iv_t^p, v_{s'}^iv_{t'}^q\in R$ for some $i,p,q\in\{1,\ldots,r\}$ and $s,s',t,t'\in\{1,\ldots,n\}$, then $s=s'$.
\end{claim}

\begin{proof}[Proof of Claim~\ref{cl:third}]
To obtain a contradiction, assume that there are $i,p,q\in\{1,\ldots,r\}$ and $s,s',t,t'\in\{1,\ldots,n\}$ such that $s\neq s'$ and
$e=v_s^iv_t^p$ and $e'=v_{s'}^iv_{t'}^q$ are edges of $R$.  Note that by Claim~\ref{cl:first}, $p\neq q$.
 Let 
$\mathcal{I}=\{1,\ldots,r\}\cup\{e,e'\}$ and $\mathcal{J}=\{i\}\cup\{1+r,\ldots,2r\}$. We again show that $\rank(\hat{L}[\mathcal{I},\mathcal{J}])\geq r+1$
using Gaussian elimination combined with Lemma~\ref{lem:diag}.
We write  $\hat{L}[\mathcal{I},\mathcal{J}]$ as follows assuming that $i<p<q$ (the other cases are symmetric):
{\small
\[
\left(
\begin{array}{c|ccc|c|ccc|c|ccc|c|ccc}
\delta_{1i}   & a &\cdots &0&0&0&\cdots  &0    &0       &0&\cdots&0&0&0&\cdots&0\\
\vdots          &    &\ddots&  &  \vdots &  &\ddots &     &\vdots&  &\ddots& &\vdots& &\ddots &  \\
\delta_{i-1i}& 0 &  \cdots&a&0&0&\cdots&0&0&0&\cdots&0&0&0&\cdots&0\\ 
\hline
a+\delta_{ii}& 0 & \cdots&0&a&0&\cdots&0&0&0&\cdots&0&0&0&\cdots&0\\ 
\hline
\delta_{i+1i}   & 0 &\cdots &0&0&a&\cdots  &0    &0       &0&\cdots&0&0&0&\cdots&0\\
\vdots          &    &\ddots&  &  \vdots &  &\ddots &     &\vdots&  &\ddots& &\vdots& &\ddots &  \\
\delta_{p-1i}& 0 &  \cdots&0&0&0&\cdots&a&0&0&\cdots&0&0&0&\cdots&0\\ 
\hline
\delta_{pi}& 0 & \cdots&0&0&0&\cdots&0&a&0&\cdots&0&0&0&\cdots&0\\ 
\hline
\delta_{p+1i}   & 0 &\cdots &0&0&0&\cdots  &0    &0       &a&\cdots&0&0&0&\cdots&0\\
\vdots          &    &\ddots&  &  \vdots &  &\ddots &     &\vdots&  &\ddots& &\vdots& &\ddots &  \\
\delta_{q-1i}& 0 &  \cdots&0&0&0&\cdots&0&0&0&\cdots&a&0&0&\cdots&0\\ 
\hline
\delta_{qi}& 0 & \cdots&0&0&0&\cdots&0&0&0&\cdots&0&a&0&\cdots&0\\ 
\hline
\delta_{q+1i}   & 0 &\cdots &0&0&0&\cdots  &0    &0       &0&\cdots&0&0&a&\cdots&0\\
\vdots          &    &\ddots&  &  \vdots &  &\ddots &     &\vdots&  &\ddots& &\vdots& &\ddots &  \\
\delta_{ri}& 0 &  \cdots&0&0&0&\cdots&0&0&0&\cdots&0&0&0&\cdots&a\\ 
\hline
c(a-s)+\xi_{ei}& \xi_{e1}' & \cdots& & & &\cdots &\xi_{ep-1}'&ca+\xi_{ep}'&\xi_{ep+1}'&\cdots&\xi_{eq-1}'&\xi_{eq}'&\xi_{eq+1}'&\cdots&\xi_{er}'\\
c(a-s')+\xi_{e'i}& \xi_{e'1}' & \cdots& & & &\cdots &\xi_{e'p-1}'&\xi_{e'p}'&\xi_{e'p+1}'&\cdots&\xi_{e'q-1}'&ca+\xi_{e'q}'&\xi_{e'q+1}'&\cdots&\xi_{e'r}'\\
\end{array}
\right).
\]}
We subtract the $p$-th row and the $q$-th multiplied by $c$ from the second last and the last row
respectively and obtain the following matrix:
{\small
\[
\left(
\begin{array}{c|ccc|c|ccc}
\delta_{1i}   & a &\cdots &0&0&0&\cdots  &0   \\
\vdots          &    &\ddots&  &  \vdots &  &\ddots & \\
\delta_{i-1i}& 0 &  \cdots&a&0&0&\cdots&0\\ 
\hline
a+\delta_{ii}& 0 & \cdots&0&a&0&\cdots&0\\ 
\hline
\delta_{i+1i}   & 0 &\cdots &0&0&a&\cdots  &0 \\
\vdots          &    &\ddots&  &  \vdots &  &\ddots &    \\
\delta_{ri}& 0 &  \cdots&0&0&0&\cdots&a\\ 
\hline
c(a-s)+\xi_{ei}-c\delta_{pi}& \xi_{e1}' & \cdots&\xi_{ei-1}' & \xi_{ei}'& \xi_{ei+1}'&\cdots &\xi_{er}'\\
c(a-s')+\xi_{e'i}-c\delta_{qi}& \xi_{e'1}' & \cdots&\xi_{e'i-1}'  &\xi_{e'i}'  & \xi_{e'i+1}' &\cdots &\xi_{e'r}'
\end{array}
\right).
\]}
Let $\alpha_h=\xi_{eh}'-\xi_{e'h}'$ for $h\in\{1,\ldots,r\}$.
We subtract the last row from the previous and delete the last row:
{\small
\[
\left(
\begin{array}{c|ccc|c|ccc}
\delta_{1i}   & a &\cdots &0&0&0&\cdots  &0   \\
\vdots          &    &\ddots&  &  \vdots &  &\ddots & \\
\delta_{i-1i}& 0 &  \cdots&a&0&0&\cdots&0\\ 
\hline
a+\delta_{ii}& 0 & \cdots&0&a&0&\cdots&0\\ 
\hline
\delta_{i+1i}   &0 &\cdots &0&0&a&\cdots  &0 \\
\vdots          &    &\ddots&  &  \vdots &  &\ddots &    \\
\delta_{ri}& 0 &  \cdots&0&0&0&\cdots&a\\ 
\hline
c(s'-s)+\xi_{ei}-\xi_{e'i} -c(\delta_{pi}-\delta_{qi})& \alpha_1 & \cdots&\alpha_{i-1} & \alpha_i& \alpha_{i+1}&\cdots &\alpha_r\\
\end{array}
\right).
\]}
By Claim~\ref{cl:second}, $|\delta_{pi}-\delta_{qi}|<1/2$. Hence, $|c(s'-s)+\xi_{ei}-\xi_{e'i} -c(\delta_{pi}-\delta_{qi})|>0$ as $|\xi_{ei}-\xi_{e'i}|<2\sqrt{D'}$. We let 
$\beta=a/(c(s'-s)+\xi_{ei}-\xi_{e'i} -c(\delta_{pi}-\delta_{qi}))$ and define $\beta_h=-\beta\alpha_h$ for $h\in\{1,\ldots,r\}$.
We subtract from the $i$-th row the last row multiplied by $\beta$ 
and then make the last row the first one:
{\small
\[
\left(
\begin{array}{c|ccc|c|ccc}
c(s'-s)+\xi_{ei}-\xi_{e'i} -c(\delta_{pi}-\delta_{qi})& \alpha_1 & \cdots&\alpha_{i-1} & \alpha_i& \alpha_{i+1}&\cdots &\alpha_r\\
\hline
\delta_{1i}   & a &\cdots &0&0&0&\cdots  &0   \\
\vdots          &    &\ddots&  &  \vdots &  &\ddots & \\
\delta_{i-1i}& 0 &  \cdots&a&0&0&\cdots&0\\ 
\hline
\delta_{ii}& \beta_1  & \cdots&\beta_{i-1}&a+\beta_i&\beta_{i+1}&\cdots&\beta_r\\ 
\hline
\delta_{i+1i}   &0 &\cdots &0&0&a&\cdots  &0 \\
\vdots          &    &\ddots&  &  \vdots &  &\ddots &    \\
\delta_{ri}& 0 &  \cdots&0&0&0&\cdots&a\\ 
\end{array}
\right).
\]}
Let $U$ denote the obtained $(r+1)\times(r+1)$ matrix.

Recall that $|\delta_{hi}|\leq 2(\sqrt{r}+1)\sqrt{D'}$  for $h\in\{1,\ldots,r\}$, $|\xi_{ei}|,|\xi_{e'i}|\leq\sqrt{D'}$ and $|\xi_{eh}'|,|\xi_{e'h}'|\leq \sqrt{D'}$ for $h\in\{1,\ldots,r\}$.
In particular, this means that $|\alpha_h|\leq 2\sqrt{D'}$ for $h\in\{1,\ldots,r\}$. 
By Claim~\ref{cl:second}, $|\delta_{pi}-\delta_{qi}|<1/2$. Hence, 
$|c(s'-s)+\xi_{ei}-\xi_{e'i} -c(\delta_{pi}-\delta_{qi})|>|c/2+\xi_{ei}-\xi_{e'i}|>a$, because  $|\xi_{ei}-\xi_{e'i}|<2\sqrt{D'}$. 
Therefore, $|\beta|<1$ and $|\beta_h|\leq 2\sqrt{D'}$ for $h\in\{1,\ldots,r\}$.
We also have that 
$|c(s'-s)-c(\delta_{pi}-\delta_{qi})|>a$.
 Then, by Lemma~\ref{lem:diag}, $\rank(U)$ is the rank of the $(r+1)\times(r+1)$ diagonal matrix
{
 \[
\left(
\begin{array}{c|ccc}
c(s'-s)-c(\delta_{pi}-\delta_{qi})& 0& \cdots&0 \\
\hline
0   & a &\cdots &0\\
\vdots &    & \ddots&   \\
0  &0&\cdots&a
\end{array}
\right).
\]}
That is, $\rank(U)=r+1\leq \rank (\hat{L}[\mathcal{I},\mathcal{J}])\leq \rank(\hat{L})=r$, which is a contradiction.
\end{proof}

Now we prove Claim~\ref{cl:clique}.
By Claim~\ref{cl:first},  for each $i\in\{1,\ldots,r\}$, $R$ contains exactly $r-1$ edges incident to vertices of $V_i$. Then by Claim~\ref{cl:third},  these vertices are incident to the same vertex of $V_i$. Denote these veritices by $v_{i_1}^1,\ldots,v_{i_r}^r$. Then $R=\{v_{i_s}^sv_{i_t}^t\mid 1\leq s<t\leq r\}$. This completes the proof of Claim~\ref{cl:clique}. 

\medskip
We obtain that $(G,V_1,\ldots,V_r)$ is a yes-instance of \textsc{Multicolored Clique}. This proves Claim~\ref{cl:crucial}.

\medskip
To complete the proof of Theorem~\ref{thm:hard-dim}, recall that $M$ is $(r+1)m\times 2r$ integer matrix and the absolute value of each element is at most $c=\Oh(r^2mn^2)$. Therefore, the bitsize $N$ of $M$ is $\Oh(|V(G)|^4\log |V(G)|)$. Observe that, given $(G,V_1,\ldots,V_r)$, $M$ can be constructed in polynomial time.  
Assume that there is a $\omega$-approximation algorithm $\mathcal{A}$ for \probPCAOut with running time $f(d)\cdot N^{o(d)}$ for a computable function $f$.
If $(G,V_1,\ldots,V_r)$ is a yes-instance of \textsc{Multicolored Clique}, then ${\sf Opt}(G,r,k)\leq D$. Therefore, $\mathcal{A}$ applied to $(M,r,k)$ reports that there is  a feasible solution $(L,S)$ with $\|M-L-S\|_F^2\leq \omega{\sf Opt}(G,r,k)\leq \omega D$ by Claim~\ref{cl:crucial}. For the opposite direction, if $\mathcal{A}$ reports that there is a feasible solution $(L,S)$ with $\|M-L-S\|_F^2\leq \omega D$, then $(G,V_1,\ldots,V_r)$ is a yes-instance of \textsc{Multicolored Clique} by Claim~\ref{cl:crucial}. Hence, $\mathcal{A}$ reports the existence of a feasible  
solution $(L,S)$ with $\|M-L-S\|_F^2\leq \omega D$ if and only if $(G,V_1,\ldots,V_r)$ is a yes-instance of \textsc{Multicolored Clique}.
Since $N=\Oh(|V(G)|^4\log |V(G)|)$ and $2r$, we obtain that $\mathcal{A}$ solves \textsc{Multicolored Clique} in time $f(2k)\cdot |V(G)|^{o(r)}$ contradicting ETH. 
\end{proof}

As \textsc{Multicolored Clique}  is well-known to be \classW{1}-hard (see \cite{DBLP:journals/tcs/FellowsHRV09,CyganFKLMPPS15}), our reduction gives the following corollary based on the weaker conjecture that $\classFPT\neq\classW{1}$.  We refer to the book~\cite{CyganFKLMPPS15} for the formal definitions of the parameterized complexity classes \classFPT and \classW{1}. Note that ETH implies that $\classFPT\neq\classW{1}$ but not the other way around.

\begin{corollary}\label{cor:whard-dim}
For any $\omega\geq 1$, there is no $\omega$-approximation algorithm for \probPCAOut with running time $f(d)\cdot N^{\Oh(1)}$ for any computable function $f$ unless $\classFPT=\classW{1}$, where $N$ is the bitsize of the input matrix $M$.
\end{corollary}

\fi




\iffull
\bibliographystyle{siam}
\fi
\bibliography{pca_with_outliers,k-clustering}

\begin{thebibliography}{10}

\bibitem{BasuPR06}
{\sc S.~Basu, R.~Pollack, and M.-F. Roy}, {\em Algorithms in Real Algebraic
  Geometry (Algorithms and Computation in Mathematics)}, Springer-Verlag,
  Berlin, Heidelberg, 2006.

\bibitem{BlumHK17}
{\sc A.~Blum, J.~Hopcroft, and R.~Kannan}, {\em Foundations of Data Science},
  June 2017.

\bibitem{bouwmans2016handbook}
{\sc T.~Bouwmans, N.~S. Aybat, and E.-h. Zahzah}, {\em Handbook of robust
  low-rank and sparse matrix decomposition: Applications in image and video
  processing}, Chapman and Hall/CRC, 2016.

\bibitem{CandesLMW11}
{\sc E.~J. Cand{\`{e}}s, X.~Li, Y.~Ma, and J.~Wright}, {\em Robust principal
  component analysis?}, J. {ACM}, 58 (2011), pp.~11:1--11:37.

\bibitem{Caruso06}
{\sc F.~Caruso}, {\em The {SARAG} library: Some algorithms in real algebraic
  geometry}, in Mathematical Software - ICMS 2006, A.~Iglesias and N.~Takayama,
  eds., Berlin, Heidelberg, 2006, Springer Berlin Heidelberg, pp.~122--131.

\bibitem{ChandrasekaranSPW11}
{\sc V.~Chandrasekaran, S.~Sanghavi, P.~A. Parrilo, and A.~S. Willsky}, {\em
  Rank-sparsity incoherence for matrix decomposition}, {SIAM} Journal on
  Optimization, 21 (2011), pp.~572--596.

\bibitem{chen2011robust}
{\sc Y.~Chen, H.~Xu, C.~Caramanis, and S.~Sanghavi}, {\em Robust matrix
  completion and corrupted columns}, in Proceedings of the 28th International
  Conference on Machine Learning (ICML), 2011, pp.~873--880.

\bibitem{CyganFKLMPPS15}
{\sc M.~Cygan, F.~V. Fomin, L.~Kowalik, D.~Lokshtanov, D.~Marx, M.~Pilipczuk,
  M.~Pilipczuk, and S.~Saurabh}, {\em Parameterized Algorithms}, Springer,
  2015.

\bibitem{eckart1936approximation}
{\sc C.~Eckart and G.~Young}, {\em The approximation of one matrix by another
  of lower rank}, Psychometrika, 1 (1936), pp.~211--218.

\bibitem{DBLP:journals/tcs/FellowsHRV09}
{\sc M.~R. Fellows, D.~Hermelin, F.~A. Rosamond, and S.~Vialette}, {\em On the
  parameterized complexity of multiple-interval graph problems}, Theoretical
  Computer Science, 410 (2009), pp.~53--61.

\bibitem{FominLMSZ18}
{\sc F.~V. Fomin, D.~Lokshtanov, S.~M. Meesum, S.~Saurabh, and M.~Zehavi}, {\em
  Matrix rigidity from the viewpoint of parameterized complexity}, {SIAM} J.
  Discrete Math., 32 (2018), pp.~966--985.

\bibitem{GillisV18}
{\sc N.~Gillis and S.~A. Vavasis}, {\em On the complexity of robust {PCA} and
  $\ell_1$-norm low-rank matrix approximation}, Math. Oper. Res., 43 (2018),
  pp.~1072--1084.

\bibitem{grigoRigid}
{\sc D.~Grigoriev}, {\em Using the notions of separability and independence for
  proving the lower bounds on the circuit complexity (in russian)}, Notes of
  the Leningrad branch of the Steklov Mathematical Institute, Nauka,  (1976).

\bibitem{grigoRigidTra}
\leavevmode\vrule height 2pt depth -1.6pt width 23pt, {\em Using the notions of
  separability and independence for proving the lower bounds on the circuit
  complexity}, Journal of Soviet Math., 14 (1980), pp.~1450--1456.

\bibitem{hotelling1933analysis}
{\sc H.~Hotelling}, {\em Analysis of a complex of statistical variables into
  principal components.}, Journal of educational psychology, 24 (1933), p.~417.

\bibitem{ImpagliazzoPZ01}
{\sc R.~Impagliazzo, R.~Paturi, and F.~Zane}, {\em Which problems have strongly
  exponential complexity}, J. Computer and System Sciences, 63 (2001),
  pp.~512--530.

\bibitem{LokshtanovMS11}
{\sc D.~Lokshtanov, D.~Marx, and S.~Saurabh}, {\em Lower bounds based on the
  exponential time hypothesis}, Bulletin of the {EATCS}, 105 (2011),
  pp.~41--72.

\bibitem{pearson1901liii}
{\sc K.~Pearson}, {\em Liii. on lines and planes of closest fit to systems of
  points in space}, The London, Edinburgh, and Dublin Philosophical Magazine
  and Journal of Science, 2 (1901), pp.~559--572.

\bibitem{Shlens14}
{\sc J.~Shlens}, {\em A tutorial on principal component analysis}, CoRR,
  abs/1404.1100 (2014).

\bibitem{Valiant1977}
{\sc L.~G. Valiant}, {\em Graph-theoretic arguments in low-level complexity},
  in MFCS, 1977, pp.~162--176.

\bibitem{VaswaniN18}
{\sc N.~Vaswani and P.~Narayanamurthy}, {\em Static and dynamic robust {PCA}
  and matrix completion: {A} review}, Proceedings of the {IEEE}, 106 (2018),
  pp.~1359--1379.

\bibitem{Weyl1912}
{\sc H.~Weyl}, {\em Das asymptotische {V}erteilungsgesetz der {E}igenwerte
  linearer partieller {D}ifferentialgleichungen (mit einer {A}nwendung auf die
  {T}heorie der {H}ohlraumstrahlung)}, Math. Ann., 71 (1912), pp.~441--479.

\bibitem{WrightGRPM09}
{\sc J.~Wright, A.~Ganesh, S.~R. Rao, Y.~Peng, and Y.~Ma}, {\em Robust
  principal component analysis: Exact recovery of corrupted low-rank matrices
  via convex optimization}, in Proceedings of 23rd Annual Conference on Neural
  Information Processing Systems (NIPS), Curran Associates, Inc., 2009,
  pp.~2080--2088.

\bibitem{XuCS10}
{\sc H.~Xu, C.~Caramanis, and S.~Sanghavi}, {\em Robust {PCA} via outlier
  pursuit}, in Proceedings of the 24th Annual Conference on Neural Information
  Processing Systems (NIPS), Curran Associates, Inc., 2010, pp.~2496--2504.

\end{thebibliography}

\end{document}